\documentclass[conference]{IEEEtran}
\newcount\DraftStatus  
\usepackage{comment}
\excludecomment{JournalOnly}  
\includecomment{ConferenceOnly}  
\excludecomment{TulipStyle}
\usepackage{comment}
\usepackage[draft,colorinlistoftodos,color=orange!30]{todonotes}
\usepackage{color}

\usepackage{graphicx}
\graphicspath{{./figures/}}

\newtheorem{assumption}{Assumption}

\usepackage{graphicx}
\usepackage{algorithm}
\usepackage[noend]{algpseudocode}
\usepackage{amsmath}
\usepackage{amssymb}
\renewcommand{\algorithmicrequire}{\textbf{Blockchain process:}}
\renewcommand{\algorithmicensure}{\textbf{Federated learning process:}}
\usepackage{breqn}
\usepackage{subcaption}
\usepackage{multirow}
\usepackage{psfrag}
\usepackage{url}
\usepackage{hyperref}
\usepackage{stmaryrd}
\usepackage{cleveref}
\crefname{assumption}{assumption}{assumptions}
\usepackage{booktabs}
\usepackage{rotating}
\usepackage{colortbl}
\usepackage{paralist}
\usepackage{epstopdf}
\usepackage{nag}
\usepackage{microtype}
\usepackage{siunitx}
\usepackage{nicefrac}
\usepackage{lipsum}
\usepackage[english]{babel}
\usepackage[pangram]{blindtext}
\usepackage{tikz}
\usepackage{tabularx}
\usetikzlibrary{fit,positioning,arrows.meta,shapes,arrows}
\usepackage[numbers]{natbib}
\usepackage{mathtools}
\numberwithin{equation}{section}

\newtheorem{theorem}{Theorem}
\newtheorem{lemma}{Lemma}

\newtheorem{definition}{Definition}

\IEEEoverridecommandlockouts
\def\BibTeX{{\rm B\kern-.05em{\sc i\kern-.025em b}\kern-.08em
    T\kern-.1667em\lower.7ex\hbox{E}\kern-.125emX}}

\begin{document}
%
%
\title{Post Quantum Secure Blockchain-based Federated Learning for Mobile Edge Computing}

\author{\IEEEauthorblockN{Rongxin Xu}
\IEEEauthorblockA{\textit{Business School} \\
\textit{Hunan University}\\
Changsha 410082, China \\
rongxinxu@hnu.edu.cn}
\and
\IEEEauthorblockN{Shiva Raj Pokhrel}
\IEEEauthorblockA{\textit{School of Information Technology} \\
	\textit{Deakin University}\\
	Geelong, Australia \\
	shiva.pokhrel@deakin.edu.au}
\and
\IEEEauthorblockN{Qiujun Lan}
\IEEEauthorblockA{Business School \\
	\textit{Hunan University}\\
	Changsha 410082, China \\
	lanqiujun@hnu.edu.cn}
\and
\IEEEauthorblockN{Gang Li}
\IEEEauthorblockA{\textit{School of Information Technology} \\
\textit{Deakin University}\\
Geelong, Australia\\
gang.li@deakin.edu.au}
}

\maketitle

\begin{abstract}
	Mobile Edge Computing (MEC) has been a promising paradigm for communicating and edge processing of data on the move. We aim to employ  Federated Learning (FL) and prominent features of blockchain into MEC architecture such as connected autonomous vehicles to enable complete decentralization, immutability, and rewarding mechanisms simultaneously. FL is advantageous for mobile devices with constrained connectivity since it requires model updates to be delivered to a central point instead of substantial amounts of data communication. For instance, FL in autonomous, connected vehicles can increase data diversity and allow model customization, and predictions are possible even when the vehicles are not connected (by exploiting their local models) for short times. However, existing synchronous FL and Blockchain incur extremely high communication costs due to mobility-induced impairments and do not apply directly to MEC networks.  We propose a fully asynchronous Blockchained Federated Learning (BFL) framework referred to as BFL-MEC, in which the mobile clients and their models evolve independently yet guarantee stability in the global learning process. More importantly, we employ \textit{post-quantum secure} features over BFL-MEC to verify the client's identity and defend against malicious attacks. All of our design assumptions and results are evaluated with extensive simulations.
\end{abstract}



Federated learning (FL) has emerged as a 
promising approach to tackle the limitations of traditional centralized machine learning (ML) methods~\cite{mcmahanCommunicationEfficientLearningDeep2017,yangFLASHHeterogeneityAwareFederated2022}, 
which are unable to handle the challenges posed by big data and complex models. 
By keeping the raw data on the end devices, 
also known as \textit{clients}~\cite{pokhrel_federated_2020}, 
FL ensures data privacy and ownership. 
The system uses a centralized global server to aggregate updates from clients, 
and iteratively distributes the new global models back to the clients. 
However, the traditional FL setup based on a centralized server has several drawbacks, 
including \emph{single point of failure} and instability~\cite{romanFeaturesChallengesSecurity2013}. 
Furthermore, 
FL is vulnerable to attacks by malicious or compromised clients, 
who can modify local gradients and launch inference attacks~\cite{nasr_comprehensive_2019}. 
These attacks can lead to model poisoning, 
which undermines the integrity and reliability of the whole system.

Blockchain and FL are two distributed technologies 
that can complement each other to create a secure, transparent, 
and incentive-compatible system for decentralized learning and distributed computing~\cite{nguyen_federated_2021, pokhrel_federated_2020}. 
The inherent properties of blockchain, 
such as immutability and traceability, 
offer several benefits to FL, 
making it an attractive solution for decentralized machine learning. 
By incorporating blockchain into FL, 
it is possible to create a robust, intelligent, 
and privacy-preserving framework. 
\emph{Blockchain-based Federated Learning} (BFL)~\cite{pokhrel_federated_2020, Xu2023} 
is one such solution that leverages 
the benefits of both blockchain and FL to enable decentralized learning. 
In BFL, local updates and global models recorded on the blockchain 
to ensure security and promote traceability, 
and clients can seamlessly acquire new global parameters 
through a consensus mechanism. 
This approach enables the system to provide an incentive mechanism for 
clients and creates a transparent and trustworthy environment 
for distributed machine learning.

More importantly, 
recent advances in edge computing ensure good communication 
and data transmission by enabling offloading of computing tasks to 
edge nodes with close proximity to the clients~\cite{shiEdgeComputingVision2016,sunAcceleratingConvergenceFederated2023}. 
Applying BFL to mobile edge computing (MEC) networks can shape 
stunning future paradigms for the 6G era. 
Examples include autonomous vehicle networks, 
mobile crowdsensing and metaverse. 

However, with the advent of quantum computing, 
security has become a significant concern. 
Even the most secure encryption methods currently 
in use can be broken by quantum computers, 
which makes the confidentiality of data transmission in BFL vulnerable. 
For instance, 
BFL typically uses the RSA algorithm to sign and encrypt the uploaded local gradients, 
which is aimed at protecting data from being tampered with or forged in transmission. 
Nevertheless, 
a malicious attacker equipped with a quantum computer 
could undermine the security of BFL by attacking the RSA algorithm 
to gain access to plaintext or even forge signatures, 
thereby launching severe privacy and model attacks. 
Therefore, 
to counter the increasing quantum threat, 
we are committed to developing new BFL-MEC systems 
that employ post-quantum secure design to ensure 
that quantum computing algorithms such as 
Shor's cannot efficiently solve the 
cryptographic problems used in our system.

In addition,
MEC systems require a completely asynchronous design~\cite{liuAdaptiveAsynchronousFederated2023} 
because it is not realistic for a constantly moving client to maintain 
a long communication with the edge nodes, 
in contrast to FL's synchronization assumption~\cite{Xu2023}. 
Additionally, 
while blockchain rewards miners that successfully mined the block, 
this approach does not incentivize clients to make significant contributions 
to global updates in BFL~\cite{wittDecentralIncentivizedFederated2023}. 
As a result, 
a unique incentive mechanism is necessary to attract potential participants 
and retain clients who make substantial contributions. 
However, 
current works have not fully explored the design 
and implementation of an incentive mechanism in BFL, 
particularly in an asynchronous setting, 
also known as vanilla BFL, 
which has limited its practical applications in MEC networks~\cite{zhangScalableLowLatencyFederated2022}.
Specifically, 
we summarize the challenges of applying vanilla BFL 
to MEC networks as follows.
\begin{itemize}
	\item \textbf{Fully asynchronous design.}
	FL has a periodic 
	learning-updating-waiting process 
	while the blockchain keeps running. 
	In vanilla BFL~\cite{pokhrel_federated_2020}, 
	the blockchain is included in the scope of the synchronicity assumption, 
	which means the working state of the blockchain will 
	exactly match the learning process of FL~\cite{Xu2023}. 
	For example, 
	the blockchain will not run until the learning process starts. 
	Such an assumption is clearly unrealistic and also undermines 
	the security basis of the blockchain, 
	that is, 
	the competition between computing nodes.
	At the same time, 
	the nature of MEC leads to the possibility that clients may 
	quickly lose communication with an edge node. 
	Overcoming this issue is a challenging undertaking 
	because discarding some of the local updates can 
	have a negative impact on global learning. 
	This problem can become even more complex when 
	the local updates recorded in the blockchain 
	fail to reflect the actual FL stage, 
	resulting in empty blocks~\cite{bao_flchain_2019}.
	Therefore, 
	it is crucial to rethink BFL for a greenfield asynchronous design, 
	especially given the loss of communication in the MEC network.
	As a result, redeveloping the BFL to harness fully asynchronous aspects is a non-trivial task, 
	especially given the problem of communication loss in the MEC network. 
	
	\item \textbf{Contribution-based incentive.}
	One of the benefits of blockchain in FL is its ability to provide incentives, 
	as it rewards nodes that successfully mine blocks. 
	However, BFL requires a different incentive mechanism to reward clients 
	who contribute more to the global aggregation, 
	especially in data-intensive tasks~\cite{sunPainFLPersonalizedPrivacyPreserving2021}. 
	The challenge is to differentiate client contributions 
	without relying on self-reported contributions, 
	as this could lead to cheating by clients. 
	BFL cannot easily verify clients' dishonesty due to 
	limitations in checking raw data. 
	Vanilla BFL mainly relies on self-reported contributions 
	or raw data checks to determine rewards, 
	which is insufficient for a secure and 
	reliable incentive mechanism.
	
	\item \textbf{Privacy-Transparency trade-off.}
	To ensure privacy, 
	it is crucial to avoid making any data that might reveal 
	sensitive information public to all nodes in the blockchain network. 
	In the case of vanilla BFL, 
	recording all data to complete the same round of learning can 
	lead to the generation of more blocks. However, 
	since the block size is limited, it is not feasible to record 
	all data in the block, 
	as this can lead to large blocks and increase transfer time. 
	As a result, 
	the delay of vanilla BFL can be significantly high. 
	Therefore, 
	careful consideration should be given to the data 
	recorded in the block to reduce the delay of BFL.
	
	\item \textbf{Malicious attack.} 
	In real-world scenarios, malicious attacks are everywhere. 
	There are various kinds of malicious attacks against BFL. 
	For example, 
	nodes curious about the data of other participants may send 
	carefully modified local gradients and then observe changes 
	in global gradients to launch membership inference attacks. 
	Attackers may send fake local gradients by controlling clients 
	or intercepting their communication to cause global model poisoning~\cite{luoSVFLEfficientSecure2022}. 
	While it is possible to detect data that has been tampered 
	with during transmission through digital signature schemes, 
	existing schemes are no longer secure under the threat of quantum computers. 
	Therefore, 
	it is crucial to employ a quantum-secure signature scheme 
	while developing a mechanism that minimizes the impact of attacks 
	that cannot be verified by the signature scheme 
	(eg, the behavior of the participants themselves).
\end{itemize}

Hence, 
in order to implement BFL in a practical way, 
we need to consider the asynchronous design and efficient incentive mechanism. 
To this end, 
we improve the vanilla BFL framework by 
introducing a loosely coupled architecture, mobility tolerance, 
contribution-based incentive mechanism, and quantum-secure signature scheme. 
These challenges have inspired us to develop \emph{BFL-MEC}, 
that is designed for MEC networks. 
\emph{BFL-MEC} solves the problems present in vanilla BFL 
and provides a fully asynchronous and incentive-based redesign 
that enhances the framework's performance and security.

In summary, our proposed BFL-MEC framework presents 
significant contributions to the development of BFL in MEC networks, 
which are detailed as follows:
\begin{enumerate}
	\item We develop a fully asynchronous blockchain-based federated learning 
	framework by designing separate working strategies for edge nodes and clients, 
	thus, 
	it can tolerate communication loss due to mobility and 
	is suitable for MEC network and its future paradigm.
	
	\item We propose a novel incentive mechanism that evaluates 
	client contributions using clustering algorithms 
	without access to raw data. 
	This can further resist malicious attacks by discarding strategy 
	and distributing personalized rewards based on contributions, 
	thus motivating honest clients to keep contributing to the learning process.
	
	\item 
	We propose a novel weighted aggregation method in \emph{BFL-MEC}, 
	which assigns weights to clients based on their contributions. 
	This method ensures faster convergence under arbitrary data distribution 
	and minimizes the impact of malicious behaviors from the clients themselves.
	
	\item We employ a signature scheme based on post-quantum cryptography, 
	which resists the threat of quantum computers 
	while providing faster signing and verifying speeds.
\end{enumerate}

This paper is structured as follows: 
We first provide a brief overview of BFL, 
quantum computer threats ,and related works in \Cref{sec-overview}. 
Next, in \Cref{sec-modeling}, 
we present our proposed \emph{BFL-MEC} framework and 
highlight the novel insights that address the aforementioned challenges. 
We then discuss the performance of \emph{BFL-MEC} from both 
the client and edge node perspectives in \Cref{sec-summary}. 
In \Cref{sec-validation}, 
we present the experimental results that demonstrate 
the performance, latency, and security of \emph{BFL-MEC}. 
Finally, we conclude our work and discuss 
future research directions in \Cref{sec-conclusions}.

\section{PRELIMINARIES and Related Work}\label{sec-overview}
\subsection{Blockchain and Federated Learning}
By a combination of a mining competition and a consensus process, 
nodes in a blockchain network maintain a distributed ledger 
in which verified transactions are recorded securely in units known as \emph{blocks}. 
When a new block is generated, 
it is broadcasted across the network, 
and any nodes that receive the block will immediately 
cease processing transactions.
Blockchain has the following features:
\begin{itemize}
	\item \textbf{Distributed}. 
	The blockchain consists of distributed computing nodes, 
	which may have different geographic locations and 
	communicate over the network to maintain the ledger jointly.
	
	\item \textbf{Decentralization}. 
	Instead of a single authoritative central entity providing credit, 
	blockchain relies on the complete autonomy of all nodes. 
	Nodes use encrypted electronic evidence to gain trust and resolve conflicts 
	by themselves through a well-designed consensus mechanism, 
	thus forming a peer-to-peer network.
	
	\item \textbf{Security}. 
	All nodes record each transaction and 
	resolve conflicts using the majority approval, 
	so tampering is complicated when the majority of nodes are honest. 
	Meanwhile, 
	other nodes will verify the proof of work and 
	transactions of the broadcast block, 
	and the cryptographic approach ensures that 
	any modification will be detected.
	
	\item \textbf{Transparency}. 
	The information recorded on the blockchain 
	is stored in all nodes 
	and is entirely public, 
	so any entity can access it.
	
	\item \textbf{Incentives}. 
	Blockchain rewards miners who successfully generate valid blocks 
	thus inspiring nodes to contribute and 
	to ensure secure and continuous operation of the network. 
	
	\item \textbf{Asynchrony}. 
	Each node in the blockchain works independently, 
	depending on the information received and 
	the consensus mechanism to change the active status. 
	For example, 
	miners do not start mining process simultaneously, 
	and new transactions do not occur at the same time.
\end{itemize}

FL is a distributed training approach enabling end devices 
to independently train their models and pool their 
intermediate information on a central server 
to deliver global knowledge. 
Its purpose is to eliminate \emph{data islands}
~\footnote{Think of data as a vast, 
	ever-changing ocean, 
	from which various organizations, 
	like banks, siphon off information to store it privately. 
	As a result, 
	these untouched pieces of information eventually form separate islands.} 
and reap the advantages of aggregate modeling.
In the FL process, 
at the beginning of each communication round, 
clients update their local models with their own data and compute gradients. 
These gradients are then uploaded to the central server, 
which aggregates them to compute the global update. 
The central server then returns the global update to the clients, 
who use it to update their local models independently. 
This iterative process allows FL to evolve dynamically from one round to another, 
resulting in a more accurate and efficient model.
FL has the following features:
\begin{itemize}
	\item \textbf{Distributed}. 
	The clients in federated learning are distributed and 
	perform their own learning at the end devices of the network.
	
	\item \textbf{Centralization}. 
	Federated learning relies on an authoritative central entity 
	to provide credit and collect all local gradients for global updates.
	
	\item \textbf{Privacy}. 
	The local learning happens in the clients, 
	and they upload the locally updated gradients. 
	Data never leave the local nodes.
	
	\item \textbf{Synchrony}. 
	In general, 
	all nodes in FL work based on communication rounds and 
	complete the entire learning process in one round 
	before starting the next round.
\end{itemize}

\subsection{Blockchained Federated Learning (BFL)}\label{sec-bfl}

\begin{figure}[t]
	\centering
	\includegraphics[width=0.79\linewidth]{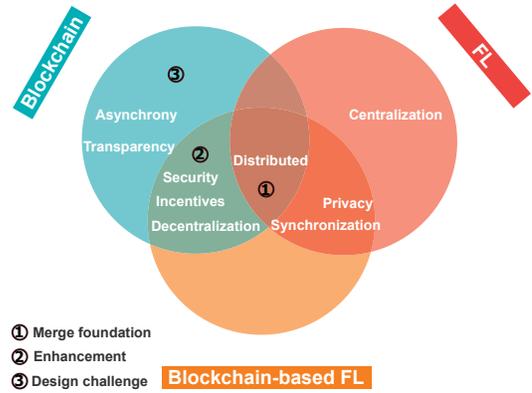}
	\caption{Opportunities and challenges of BFL}
	\label{fig:venn}
\end{figure}

The opportunities and challenges of integrating Blockchain 
and FL are visualized in \Cref{fig:venn}. 
The distributed nature of both blockchain and FL offers 
a strong foundation for their seamless integration.
Moreover, 
the decentralization feature of blockchain 
empowers FL with better autonomy. 
It is worth noting that 
blockchain, 
as a data encryption and value-driven system, 
always guarantees the security of the data exchange process in the BFL.
However, 
BFL design consists of non-trivial challenges as shown in \Cref{fig:venn}. 
As an illustration, 
the asynchronous characteristic of blockchain necessitates a thorough integration 
with the communication rounds mechanism of FL.
Also, 
there is a huge trade-off 
between the transparency of blockchain and the privacy of FL. 
We must carefully determine which kind of data should be public to avoid weakening privacy.

\subsection{Quantum Computer Threats}
In the future, quantum computers will become one of 
the primary computing infrastructures 
due to their powerful parallel computing capabilities. 
In contrast to conventional computers, 
quantum computers run quantum bits based on quantum entities 
such as photons, electrons, and ions. 
Its computational power is several orders of magnitude higher 
than that of conventional computers and shows a stunning 
increase as the number of quantum bits increases.

The competition for quantum computing has long been underway, 
and many tech giants and research institutions have already 
started building their infrastructure. 
Among them, 
IBM and Google have already surpassed 100 quantum bits in 2021-2022, 
and plan to surpass 1000 quantum bits within 2023, 
and 1 million in 2030, respectively.

The cornerstone of cryptography is mathematical puzzles 
that cannot be solved in a short time, 
whereas the advent of quantum computers makes it possible 
to perform brute force cracking in an acceptable time. 
Against this background, 
concerns about the security of current cryptographic 
systems have been raised. 
For example, 
widely used cryptographic algorithms such as RCDSA, RSA, and DSA 
have been theoretically proven to be impervious to quantum attacks. 
As a result, 
quantum computers pose a widespread threat to 
systems protected by these methods. 
For example, 
blockchains and edge computing systems 
that use RSA for signature authentication.

\subsection{Related Work}
Several notable studies have explored the integration of BFL, 
including works cited as
~\cite{pokhrel_federated_2020, Xu2023, awan_poster_2019,majeed_flchain_2019,lu_blockchain_2020,li_blockchain-based_2021}. 
To address privacy concerns, 
\cite{awan_poster_2019} proposed a variant of the \emph{Paillier} cryptosystem 
that supports homomorphic encryption and proxy re-encryption, 
while \cite{majeed_flchain_2019} implemented a BFL framework called \emph{FLchain} 
that leverages ``channels" to store local model gradients in blockchain blocks. 
Components such as \emph{Ethereum} have extended the capabilities of \emph{FLchain}, 
enabling execution of global updates.

To address centralized trust issues, 
\cite{lu_blockchain_2020} incorporated differential privacy 
into permissioned blockchains, 
and \cite{li_blockchain-based_2021} proposed a BFL framework that 
stores the global model and exchanges local updates via blockchain, 
effectively eliminating the need for a central server and mitigating privacy risks. 
These works have aimed to improve the privacy and security of the vanilla BFL framework.

However, vanilla BFL design still faces several challenges~\cite{pokhrel_federated_2020, Xu2023}. 
For example, 
the asynchronous nature of blockchain requires careful integration 
with the FL communication rounds mechanism and therefore the frameworks proposed in~\cite{pokhrel_federated_2020, Xu2023} are not practicable for MEC. In addition, in our earlier BFL frameworks~\cite{pokhrel_federated_2020, Xu2023}, 
there is little consideration for evaluating the clients' contributions, 
and those that do may raise privacy concerns
~\cite{pokhrel_federated_2020,kim_blockchained_2020}. 
As a result, 
there is a need for an approach that overcomes these 
limitations and offers improved performance and security.

In this paper, 
we propose \emph{BFL-MEC}, 
a novel approach to BFL that features a full asynchronous design, 
a fair aggregation scheme, and a contribution-based incentive mechanism. 
\emph{BFL-MEC} enhances the security and robustness of BFL for 
various MEC applications by addressing key limitations 
of the vanilla BFL framework. 
By providing an incentive mechanism for clients and implementing 
a fair aggregation scheme, 
\emph{BFL-MEC} aims to ensure that clients' contributions 
are accurately and objectively evaluated. 
Overall, \emph{BFL-MEC} offers several advantages over 
the vanilla BFL framework and provides a strong foundation 
for the integration of blockchain and federated learning.

\section{BFL-MEC Design}\label{sec-modeling}
This section outlines the proposed \emph{BFL-MEC} approach 
and provides a detailed algorithm description. 
We explain how blockchain and federated learning can 
be effectively integrated by taking advantage of their internal workings. 
\Cref{tab: notations} presents a summary of the notations used in this paper.

\begin{table}[htbp]
	\centering
	\caption{Summary of notations}
	\begin{tabular}{ll}
		\toprule
		\multicolumn{2}{c}{Notations used in this paper} \\
		\midrule
		${C_i}$    & A client in BFL-MEC with index $i$. \\
		${S_k}$    & A edge node in \emph{BFL-MEC} with index $k$. \\
		$\mathcal D_i$    & The data set held by the client $i$ at a certain moment. \\
		$\eta$ & The learning rate of the client's local model. \\
		$E$   & The number of training epochs for the client's model. \\
		$B$    & The batch size used by the clients. \\
		$n$    & The number of clients. \\
		$m$    & The number of edge nodes. \\
		$\mathcal B$    & Dataset split by batch size. \\
		$w$   & The calculated gradient. \\
		\bottomrule
	\end{tabular}%
	\label{tab: notations}%
\end{table}%

\subsection{System Overview}

\begin{figure}[htbp]
	\centering
	\includegraphics[width=0.8\linewidth]{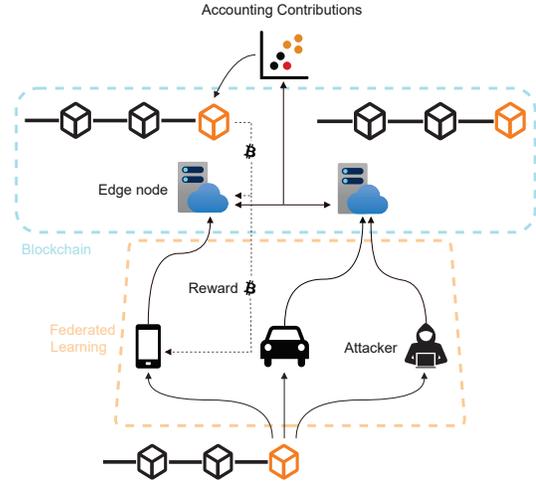}
	\caption{The framework of \emph{BFL-MEC}}
	\label{fig:bfl}
\end{figure}
\Cref{fig:bfl} provides a top-level perspective 
of the \emph{BFL-MEC} framework, 
which incorporates multiple edge nodes instead of a central server 
for federated learning. 
By doing so, 
the framework avoids the possibility of a \emph{single point of failure}.
Also, 
they play the role of miners in the blockchain, 
jointly maintaining a ledger and mining blocks.
The whole process can be seen as asynchronous communication between
$n$ clients $\{ {C_i}\} _{i = 1}^n$ 
and a set of $m$ miners 
$\{ {S_k}\} _{k = 1}^m$ 
to handle the blockchain process. 
However, 
instead of sending transaction information, 
the clients send the gradients of their local models to the edge nodes, 
who also compute the global gradients and record them on the block chain for sharing. 
The client $C_i$ move around the environment and collect new data, 
and the resulting dataset ${\mathcal D_i}$ is used to train the local model. 
Moreover, they read the global model from the blockchain to update the local model.

\subsection{Post-Quantum Secure Feature}\label{pqc}
It is essential to avoid using local gradients for 
subsequent global updates in a federated learning setting 
without proper verification. 
The risk of gradient attacks launched by malicious 
clients, which can easily forge information, 
is significant~\cite{nasr_comprehensive_2019}.
To mitigate the risk of gradient attacks in federated learning, 
clients often generate digital signatures for 
local gradients using their private keys. 
Edge nodes then use clients' public keys to verify the signatures, 
ensuring that the information has not been tampered with.
However,
The great majority of extant signature algorithms rely on classical cryptography, 
such as RSA. 
However, as aforementioned, 
these schemes are no longer robust to quantum computers. 
To this end, 
we use post-quantum cryptography (PQC) to ensure that 
the identities of both parties are verified and quantum security. 

\textbf{Lattice-based Cryptography.}
Cryptography based on lattices is safe because it is impossible 
to break with sufficient force in a reasonable amount of time, 
not even with quantum computers. 
In addition to multivariate cryptography, 
the leading possibilities for post-quantum cryptography 
also include hash-based cryptography and code-based encryption. 
We illustrate in \Cref{alg:pqc} how to provide post-quantum security 
features using lattice-based cryptography.

\begin{algorithm}[htbp]  
	\caption{Lattice-based Digital Signature Scheme}  
	\label{alg:pqc}  
	\begin{algorithmic}[1]  
		\Procedure{KeypairGen}{}
		\State $\mathbf{A} \leftarrow R_{q}^{k \times \ell}$
		\State $\left(\mathrm{s}_{1}, \mathrm{~s}_{2}\right) \leftarrow S_{\eta}^{\ell} \times S_{\eta}^{k}$
		\State $\mathbf{t}:=\mathbf{A} \mathbf{s}_{1}+\mathbf{s}_{2}$ \\
		\Return $ \left(p k=(\mathbf{A}, \mathbf{t}), s k=\left(\mathbf{A}, \mathbf{t}, \mathrm{s}_{1}, \mathbf{s}_{2}\right)\right) $
		\EndProcedure
		\\
		\Procedure{PQCSign}{$sk, M$}
		\State $ \mathbf{Z}:=\perp $
		\While{$ \mathrm{z}:=\perp $}
		\State $ \mathbf{y} \leftarrow S_{\gamma_{1}-1}^{\ell} $
		\State $ \mathbf{w}_{1}:=\operatorname{HighBits}\left(\mathbf{A y},2\gamma_{2}\right) $
		\State $ c \in B_{\tau}:=\mathrm{H}\left(M \| \mathbf{w}_{1}\right) $
		\State $ \mathbf{z}:=\mathbf{y}+c s_{1} $
		\If{$ \|\mathbf{z}\|_{\infty} \geq \gamma_{1}-\beta $ or $ \| \operatorname{LowBits} \left(\mathbf{A y}-\operatorname{cs}_{2},2\gamma_{2}\right) \|_{\infty} \geq \gamma_{2}-\beta $}
		$ \mathrm{z}:=\perp $
		\EndIf
		\EndWhile
		\\
		\Return $ \sigma=(\mathbf{z}, c) $
		\EndProcedure
		\\
		\Procedure{PQCVerify}{$ p k, M, \sigma=(\mathbf{z}, c) $}
		\State $ \mathbf{w}_{1}^{\prime}:=\operatorname{HighBits}\left(\mathbf{A z}-c \mathbf{t},2\gamma_{2}\right) $ \\
		\Return $ \llbracket\|\mathbf{z}\|_{\infty}<\gamma_{1}-\beta \rrbracket $ and $ \llbracket c=\mathrm{H}\left(M \| \mathbf{w}_{1}^{\prime}\right) \rrbracket $
		\EndProcedure
	\end{algorithmic} 
\end{algorithm} 

We employ a simplified version of \emph{Dilithium} as shown in \Cref{alg:pqc}, 
which is explained in the following (see \cite{ducasCRYSTALSDilithiumLatticeBasedDigital2018} for details).
\begin{inparaenum}[i)]
	\item \emph{Keypair Generation.}
	First, 
	a $ k \times \ell $ matrix $A$ is generated over 
	the ring $ R_{q}=\mathbb{Z}_{q}[X] /\left(X^{n}+1\right) $, 
	each of whose entries is a polynomial in $R_{q}$. 
	Then we randomly sample the private key vectors $s_1$ and $s_2$, 
	and therefore compute the second term of the public key as $t=As_1+s_2$.
	\item \emph{PQC Sign.} First, 
	we generates a masking vector of polynomials $y$ with coefficients less than $ \gamma_{1} $. 
	And $ \gamma_{1} $ is chosen strategically: 
	it is large enough that the signature does not expose the secret key,
	which means the signing process is zero-knowledge, 
	yet small enough that the signature cannot be easily forged. 
	Then, the clients compute $A_y$ and set $w_1$ to be 
	high-order bits of the coefficients in this vector, 
	where each coefficient $w$ of $A_y$, for example, 
	can be expressed canonically as $ w=w_{1} \cdot2\gamma_{2}+w_{0} $.
	The challenge $c$ is then created as the hash of the local gradients and $w_1$. 
	Finally, clients can get the signature by computing $z=y+cs_1$. 
	\item \emph{PQC Verify.} The edge nodes first computes $ \mathbf{w}_{1}^{\prime} $ 
	to be the high-order bits of $Az-ct$,
	and then accepts if all the coefficients of $z$ are less than $\gamma_{1}-\beta$ 
	and if $c$ is the hash of the message and $ \mathbf{w}_{1}^{\prime} $. 
	As $\mathrm{Az}-c \mathbf{t}=\mathrm{Ay}-c \mathrm{~s}_{2}$, we have 
	\begin{equation}
		\label{eq:pqcverify}
		\operatorname{HighBits}\left(\mathbf{A y},2\gamma_{2}\right)=\operatorname{HighBits}\left(\mathbf{A y}-c \mathbf{s}_{2},2\gamma_{2}\right).
	\end{equation}
	We know that a valid signature satisfies 
	$\|\operatorname{LowBits}\left(\mathbf{A y}-c \mathbf{s}_{2},2\gamma_{2}\right) \|_{\infty}<\gamma_{2}-\beta$, 
	and the coefficients of $cs_2$ are smaller than $\beta$, 
	and adding $cs_2$ is not enough to cause any carries by 
	increasing any low-order coefficient to have magnitude at least $\gamma_{2}$. 
	Thus the \Cref{eq:pqcverify} is true and the signature verifies correctly.
\end{inparaenum}
By using \Cref{alg:pqc} for the communication process between the clients and the edge nodes 
and replacing the RSA component of the blockchain, 
we can ensure that the whole system is post-quantum secure 
and can effectively withstand the threat of quantum computers. 
In our design, 
every client is allocated a unique private key according to their ID, 
and the corresponding public key is kept by the edge nodes. 
The information is verified using the client's public key, 
and the gradient information is signed with the private key 
to ensure that it is not tampered with, 
as demonstrated in \Cref{fig:rsa}.

\begin{figure}[htbp]
	\centering
	\includegraphics[width=0.8\linewidth]{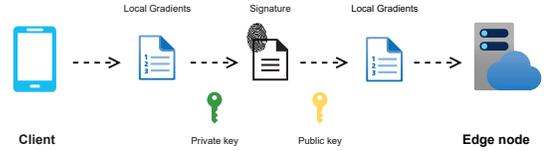}
	\caption{Miners verify transactions through PQC}
	\label{fig:rsa}
\end{figure}

\textbf{Privacy-Transparency trade-off in BFL.} 
Vanilla BFL records every local gradient in the blockchain. 
However, in the design of blockchain, 
the transaction information is public to everyone. 
In this case,
vanilla BFL is actually a white-box for the attacker, 
malicious nodes can use this information to 
perform privacy attacks and easily track the changes 
in a client's local gradient to 
launch more severe model inversion attacks~\cite{fredrikson_model_2015}. 
BFL-MEC alleviates this risk by not storing the 
original local gradients on the blockchain. 
Edge nodes, in particular, 
temporarily hold all unaggregated local gradients in their local cache pools, 
while globally aggregated local gradients are removed to free up space. 
When a new block is created, 
the edge nodes construct a transaction for each unaggregated local gradient 
that includes the sender's ID, the receiver's ID, 
and the local gradient's signature. 
After that, 
they group every transaction into a transaction set and include it in the block. 
The created global gradient and the reward list will be the first transaction in the set 
if a global update happens, in particular. 
This reduces the cost of search for the clients. 
\Cref{fig:bcrecord} depicts the structure of a block in our design.
\begin{figure}[htbp]
	\centering
	\includegraphics[width=0.56\linewidth]{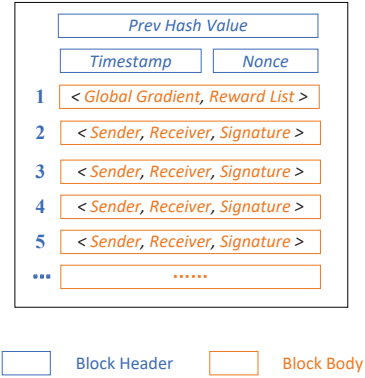}
	\caption{Privacy preserving block structure}
	\label{fig:bcrecord}
\end{figure}

\subsection{Fully Asynchronous Design}
Vanilla BFL design obeys the synchronization assumption~\cite{pokhrel_federated_2020}.
The central server samples a subset of the existing clients 
and sends instructions to these clients, which then start local updates. 
However, in an edge computing network, 
the central server cannot know precisely the current number of clients, 
nor can it guarantee that the selected clients are always in a communicable state, 
as they are constantly moving. 
In the real world, 
a fully asynchronous design is more reasonable. 
To this end, 
we design independent working strategies for clients 
and edge nodes, respectively, 
and control these two processes through the two tunable parameters $N$ and $\phi$. 

\subsubsection{Client update strategy} 
Clients keeps collecting new data from the environment and check 
if the local dataset size $\left|\mathcal{D}_{i}\right|$ exceeds the threshold $N$.
In the case of minimum constraints on the data size have been satisfied, 
the clients start executing the following work strategies.

\textbf{Anchor Jump.} 
Before beginning local training, 
clients need read the most recent global gradient from the blockchain. 
To avoid duplicate searches, 
they employ an anchor to identify the block containing the last global gradient 
and then check the first transaction of the block sequentially along this block. 
If a new global gradient is found, 
the index of the block that contains it is saved as the new anchor.

\textbf{Local model update.} 
After reading the global gradient ${w_r}$ 
from the last block of the blockchain, 
each client updates their local model accordingly. 
The data sets ${\mathcal D_i}$ are split into batches of size $B$, 
and for each epoch $i \in { 1,2,3, \ldots ,E}$, client ${C_i}$ employs 
stochastic gradient descent (SGD) to derive the gradient $w _ { r+1 } ^ { i}$, 
with the loss function $\ell$ and learning rate $\eta$, 
as outlined in \Cref{eq3}.

\begin{equation}\label{eq3}
	w^{i} \leftarrow w^{i} - 
	\eta \nabla \ell ( w^{i} ; b ) 
\end{equation}

\textbf{Uploading the gradient for mining} 
After local model update, 
the client ${C_i}$ will 
get the updated gradient $w^{i}$,
and upload it to the edge nodes for mining. 
${C_i}$ always sends its local gradients $w^{i}$ to the nearest edge node. 
However, due to constant movement, 
it may lose the connection before the upload is completed. 
In the worst case, 
it may not be able to connect with an edge node for some time, 
yet the data collected from the environment has exceeded the threshold $N$. 
In such a case, 
the client will perform local model updates normally 
while temporarily storing the not uploaded local gradients. 
When a connection can be established again, 
it computes the average of all temporarily stored local gradients as new $w^{i}$, 
Then it generates $signature$ for $w^{i}$ using the proposed PQC algorithm and finally 
sends the $w^{i}$, $signature$ and its public key $pk$ to the associated edge node. 

The whole process of client update strategy is shown in \Cref{alg:clientsupdate}.

\begin{algorithm}[htbp]  
	\caption{Clients Uptdae}  
	\label{alg:clientsupdate}  
	\begin{algorithmic}[1]  
		\If {$\left| {{\mathcal{D}_i}} \right| > N$}
		\State check the $TX_1$ of all blocks following the \textit{Anchor}
		\If {new ${w_r}$ in $block_i$}
		\State read ${w_r}$ from $block_i$ 
		\State $ \textit{Anchor} \leftarrow index(block_i) $
		\EndIf
		\State $\mathcal B \leftarrow split\ {\mathcal D_i}\ into\ batches\ of\ size\ B$ 
		\For{$each\ epoch\ i\ from\ 1\ to\ E$}
		\For{$each\ batch\ b\ \in\ \mathcal B$}
		\State $w^{i} \leftarrow w^{i } - \eta \nabla \ell ( w^ { i } ; b )$
		\EndFor
		\EndFor
		\EndIf
		\If{\{$w _ { 1 } ^ { i}, w _ { 2 } ^ { i}, ..., w _ {k} ^ { i}\}$ not uploaded}
		\State $w^{ i} \leftarrow \frac{1}{k}\sum\limits^{k} {w_{k}^i}$
		\Repeat
		\State check connection with the nearest ${S_k}$ 
		\State associate to ${S_k}$
		\State $signature \leftarrow \textbf{PQCSIGN}(w^{i}, sk)$
		\Until upload $w^i,signature,pk$ to ${S_k}$ finished
		\EndIf
	\end{algorithmic}  
\end{algorithm} 

\subsubsection{Edge nodes' update strategy} 
Before generating a block, 
the edge nodes check the distance of the latest 
one transaction $\mathcal{T}$ to the previous anchor, 
that is, 
they check how many unaggregated local gradients have been generated 
since the last global aggregation. 
If the number reaches the threshold $\phi$, 
they compute the global model and assign rewards; otherwise, 
they only record the transactions on the blockchain. 
We present here the blockchain process that must be carried out. 
As for the methods of computing the global model's and assigning rewards, 
we put them in \Cref{sec-cii,sec-convergenceproof}, respectively.

\textbf{Exchanging Gradients.}
The associated client set ${ {C_i} }$ for an edge node ${S_k}$ 
provides the updated gradient set ${ w^{i} }$. 
In parallel, 
each edge node broadcasts its own gradient set. 
${S_k}$ verifies if the received transaction already exists 
in the current gradient set ${w^{i}}$, and if not, 
it appends the new transaction. 
Eventually, all edge nodes possess the same gradient set. 
To ensure the data has not been tampered with, 
edge nodes validate the transactions from other edge nodes 
using the proposed PQC algorithm, 
as illustrated in \Cref{fig:rsa}.

\textbf{Block Mining and Consensus.}
The mining competition is a continuous process 
that involves all edge nodes. 
To be more specific, 
edge nodes continuously adjust the nonce in the block header and then 
evaluate whether the block's hash satisfies the $Target$ by utilizing SHA256. 
This entire process can be expressed as depicted in \Cref{eq5}, 
where $Targe{t_1}$ is a substantial constant that represents the highest mining difficulty. 
It's important to note that $Target$ is constant for all miners, 
and the mining difficulty is established before the algorithm commences. 
Thus, the probability of an edge node earning the right to create blocks 
is based on the speed of hash calculation.

\begin{equation}\label{eq5}
	H ( \ nonce + \ Block ) < \ Target\  = \frac{{Targe{t_1}}}{{difficulty}}
\end{equation}

Should an edge node find the solution to \Cref{eq5} before other edge nodes, 
it will promptly assemble the transaction set ${P_{i}}$ into a new block 
and disseminate this block to the network. 
Upon receipt of the new block, other edge nodes will immediately halt their 
current hash calculations and add the new block to their blockchain copies, 
subject to the block's validity being confirmed.

The whole process of edge nodes update strategy is shown in \Cref{alg:bfl}.

\begin{algorithm}[htbp]  
	\caption{Edge Nodes Update}  
	\label{alg:bfl}  
	\begin{algorithmic}[1]  
		\State do $proof\ of\ work$
		\State $W^{k}\leftarrow \{ w^{i},i = index\ of\ associate\ clients\}$
		\For {$w^i \in W^{k}$}
		\State $\textbf{PQCVERIFY}(w^i,pk)$
		\EndFor	
		\State broadcast clients updated gradient $W^{k}$
		\State received updated gradient $W^{v}$ form ${S_v}$
		\For {$w \in W^{v}$}
		\If {$w \notin W^{k}$}
		\State $W^{k}$ append $w$
		\EndIf
		\EndFor		
		\If {hash satisfies target} 
		\State $P_i\leftarrow(sender,receiver,signature)$
		\State generate and add $block(\{P_i\})$
		\State broadcast 
		\EndIf
		\If {received $bloc{k_i}$}
		\State verify $proof\ of\ work$
		\If {hash satisfies target} 
		\State stop current $proof\ of\ work$
		\State blockchain add $bloc{k_i}$
		\EndIf
		\EndIf
		\If {$\mathcal{T} - Anchor_{ts} > \phi$}
		\State $w_{g} \leftarrow \frac{1}{n}\sum\limits_{i = 1}^n {w^i},w^{i} \in {W^{k}}$ \Comment{Simple Average}
		\State $W^{k}$ append $w^{i}$
		\State $\textbf{Contribution-based\ Incentive\ Mechanism}(W^{k})$ \label{alg:cii-in-alg:bfl}
		\State $w_g \leftarrow \textbf{Fair\ Aggregation}(W^k)$ \Comment{By \Cref{eq:contriAVG}}
		\State $TX_1 \leftarrow (reward\ list, w_{g})$
		\EndIf
	\end{algorithmic}  
\end{algorithm} 

\subsection{Accounting Client's Contribution}\label{sec-cii}
After edge nodes complete the gradient exchange, 
edge node $S_k$ will have the gradient set $W^k$, 
which contains all client gradients. 
If a global update happens, 
the edge nodes first compute a temporary global gradient using simple averaging, 
and then append it to local gradient set $W^k$. 
Finally, 
the edge nodes start identifying the contributions of clients. 

\begin{algorithm}[htbp]  
	\caption{Client's Contribution Identification Algorithm}  
	\renewcommand{\algorithmicrequire}{\textbf{Input:}}
	\renewcommand{\algorithmicensure}{\textbf{Output:}}
	\label{alg:cii}  
	\begin{algorithmic}[1]  
		\Require $W^{k}$, $model\ name$, $Strategy$
		\State $Group\ List \leftarrow Clustering(model\ name, W^{k})$
		\For{$l_i \in Group\ List$} 
		\If{$w_g \in l_i$} 
		\For{$w^i \in l_i$}
		\State $\theta _{i} \leftarrow Cosine\ Distance(w^i,w_{g})$
		\State \textbf{Label} $C_i$ as high contribution
		\State \textbf{Append} $\langle C_i, \theta _i/{{\sum\limits_{k = 1}^{\lambda n} {{\theta _k}} }}*base\rangle$ to $reward\ list$
		\EndFor
		\EndIf
		\If{$w_{g} \notin l_i$}
		\ForAll{$w^i \in l_i$}
		\State \textbf{Label} $C_i$ as low contribution
		\EndFor
		\EndIf
		\EndFor
		\State $W^k \leftarrow Strategy(reward\ list, W^{k})$
		\Ensure $reward\ list$, $W^k$
	\end{algorithmic}  
\end{algorithm}

To identify contributions, 
we have implemented our method in \Cref{alg:cii}, 
which is integrated into \Cref{alg:bfl} as 
described in \Cref{alg:cii-in-alg:bfl}. 
We apply clustering algorithms to $W^k$ to 
generate various clusters of gradients, 
each cluster representing a different contribution. 
It is important to note that any clustering algorithm 
that meets the requirements can be employed here; 
nonetheless, we have chosen \emph{DBSCAN} as the default 
algorithm in our experiments due to its efficiency and simplicity. 
Clients belonging to the same cluster as the global gradient 
are regarded as having made a significant contribution 
and will be rewarded accordingly, 
whereas those who are further from the global gradient 
are considered to have made a minor contribution 
and will follow a predetermined strategy.

There are two main strategies for handling local gradients in our algorithm: 
\begin{inparaenum}[i)]
	\item keep all gradients;
	\item discard low-contributing local gradients 
	and recalculate the global updates $w_{g}$.
\end{inparaenum} 
To determine the contribution of a high-contributing client $C_i$, 
we calculate the \emph{cosine} distance $\theta_{i}$ between their local gradient 
and the global update, using this as the weight for the 
client's contribution to the global update. 
To calculate the final reward for a client, 
we multiply a $base$ reward by $\theta _i/{{\sum\limits_{k = 1}^{\lambda n} {{\theta _k}} }}$. 
We then record these rewards as key-value pairs 
$\langle C_i, \theta _i/{{\sum\limits_{k = 1}^{\lambda n} {{\theta _k}} }}*base\rangle$ 
in a $reward list$. 
When an edge node generates a new block, 
the rewards are distributed according to the reward list 
and appended to the block as transactions. 
Once blockchain consensus is achieved, 
clients receive their rewards.
The intuition behind \Cref{alg:cii} can be explained as follows: 

We explain the intuition behind \Cref{alg:cii} as follows.
\begin{itemize}
	\item \textbf{Privacy preservation.}
	In the vanilla BFL approach, 
	clients need to provide information about their 
	data dimensions to determine the rewards they will receive. 
	This incentivizes clients to cheat in order to obtain more rewards, 
	and there is no way to verify the actual data set without violating FL's guidelines. 
	In contrast, 
	the gradients can provide an intermediate representation that reflects 
	both the data size and quality. 
	By using them to perform \Cref{alg:cii}, 
	we can obtain a more objective assessment 
	that preserves the privacy of the clients.
	
	\item \textbf{Malicious attack resistance.}
	One potential threat to the security of the global model is the possibility 
	of malicious clients uploading fake local gradients. 
	However, 
	the clustering algorithm employed in \Cref{alg:cii} 
	can detect these spurious gradients since they differ 
	from the genuine ones~\cite{nasr_comprehensive_2019}. 
	By adopting the discarding strategy, 
	we can prevent these fake gradients from skewing the global model, 
	thereby preserving the security of \emph{BFL-MEC}.
\end{itemize}

We will thoroughly evaluate our contribution-based 
incentive approach in \Cref{sec-validation}.

\subsection{Fair Aggregation for Model Convergence}\label{sec-convergenceproof}
The optimization problem addressed by \emph{BFL-MEC} is 
to minimize the function $F(\mathbf{w})$, 
given by the summation of local objective functions $F_{i}(\mathbf{w})$ 
weighted by $p_i$, 
where $p_i$ is the weight of client $i$. 
This is represented by the equation 
\begin{equation*}
	\min\limits_w\left\{{F(\mathbf{w})} \triangleq \sum_{i=1}^{n} p_{i} F_{i}(\mathbf{w})\right\}. 
\end{equation*} 
For the simple average aggregation, 
where $p_{1}=p_{2}=...=p_{i}=\frac{1}{n}$, 
the global model is updated as 
\begin{equation*}\label{eq:simpleAVG}
	w_{g} \leftarrow \frac{1}{n}\sum\limits_{i = 1}^n {w^i}.
\end{equation*}

The method of simple average aggregation treats all clients' gradients 
equally and calculates their mean. 
However, 
this approach fails to account for differences 
in the sample sizes across clients, 
which can lead to unfairness. 
To address this issue, 
we use a modified approach for global gradient aggregation as follow.
\begin{equation}\label{eq:contriAVG}
	{w_{g}} \leftarrow \sum\limits_{i = 1}^n {{p_i}} w^i, \text{where}\ {p_i} = \theta _i/{{\sum\limits_{k = 1}^{\lambda n} {{\theta_k}} }}.
\end{equation}

We assign aggregation weights based on the contribution of clients 
to prevent model skew and improve accuracy, 
which addresses the issue of unequal sample sizes. 
This is more practical since each client in the mobile edge computing scenario 
has a distinct data distribution and a personalized local loss function.

The stability and convergence dynamics of \emph{BFL-MEC} can still be analyzed, 
despite using fair aggregation and an asynchronous design. 
This allows us to evaluate the performance of the algorithm 
and ensure that it achieves its intended objectives.
In our design, 
local updates and gradient uploads can be performed 
as soon as the client's local dataset size $\mathcal{D}_i$ exceeds the threshold $N$, 
which is also known as the ``activated'' state in synchronous FL. 
Therefore, 
we first define concurrency, 
that is, 
the set of clients performing local updates at each step $t$, as $C_t$. 
\begin{definition}[Concurrency]
	$ \tau_{C}^{(t)} $ is defined as the size of client set 
	for local update at step $t$, so that $ \tau_{C}^{(t)}=\left|C_{t}\right| $. 
	In consequence, we can define the maximum concurrency 
	$ \tau_{C}=\max_{t}\left\{\tau_{C}^{(t)}\right\} $
	and average concurrency 
	$ \bar{\tau}_{C}=\frac{1}{T+1} \sum_{t=0}^{T} \tau_{C}^{(t)} $ as well.
\end{definition}

Also, we define the average delay as follow.

\begin{definition}[Average Delay]
	The average delay of a client $i$ is 
	\begin{equation}
		\tau_{a v g}^{i}=\frac{1}{T_{i}}\left(\sum_{t: j_{t}=i} \tau_{t}+\sum_{k} \tau_{k}^{C_{T}, i}\right), 
	\end{equation}
	where $T_i$ is the number of times client $i$ performs local updating 
	and $\left\{\tau_{k}^{C_{T}, i}\right\}_{k}$ is the set of delays 
	from gradients of the client $i$ that are left unapplied at the last iteration.
\end{definition}

Obviously, 
the convergence of \emph{BFL-MEC} is highly relevant to $ \tau_{C}^{(t)} $. 
Moreover, 
from \cite{koloskovaSharperConvergenceGuarantees2022}, 
we can learn an essential relationship between average 
concurrency $\bar{\tau}_{C}$ and average delay $\tau_{a v g}$, 
that is, 
$ \tau_{a v g}=\frac{T+1}{T+\left|C_{T}\right|-1} \bar{\tau}_{C} \stackrel{T>\left|C_{T}\right|}{=} \mathcal{O}\left(\bar{\tau}_{C}\right) $. 
In this work, we also employ this observation to help our proof. 
For the sake of tractability, 
we have adopted the following four widely used assumptions 
from the literature~\cite{stich_local_2019,li_convergence_2020,li_federated_2020,dinhFederatedLearningWireless2021}.

\begin{assumption}[$L$-smooth~\cite{stich_local_2019,li_federated_2020,dinhFederatedLearningWireless2021}]\label{ass-smooth}
	Consider 
	${F_i}(w) \triangleq \frac{1}{n}\sum\limits_{i = 1}^n \ell  
	\left( {w;{b_i}} \right)$ 
	and ${F_i}$ is L-smooth, 
	then for all $\mathbf{v}$ and $ \mathbf{w}$,
	\[ F_{i}(\mathbf{v}) \leq F_{i}(\mathbf{w})+(\mathbf{v}- \mathbf{w})^{T} 
	\nabla F_{i}(\mathbf{w})+\frac{L}{2}\|\mathbf{v}-\mathbf{w}\|_{2}^{2}. \]
\end{assumption}

\begin{assumption}[$\mu$-strongly~\cite{stich_local_2019,li_federated_2020,dinhFederatedLearningWireless2021}]\label{ass-ustrongly}
	${F_i}$ is u-strongly convex, 
	for all $ \mathbf{v} $ and $ \mathbf{w}$,
	\[F_{i}(\mathbf{v}) \geq F_{i}(\mathbf{w})+(\mathbf{v}- \mathbf{w})^{T} 
	\nabla F_{i}(\mathbf{w})+\frac{\mu}{2}\|\mathbf{v}-\mathbf{w}\|_{2}^{2}.\]
\end{assumption}

\begin{assumption}[bounded variance~\cite{li_convergence_2020}]\label{ass-gradientbound}
	The variance of stochastic gradients in each client is bounded by: 
	\[ \mathbb{E}\left\|\nabla F_{i}
	\left(\mathbf{w}_{r}^{i}, b_i\right)-\nabla F_{i}
	\left(\mathbf{w}_{r}^{i}\right)
	\right\|^{2} \leq \sigma_{i}^{2} \]
\end{assumption}

\begin{assumption}[bounded stochastic gradient~\cite{li_convergence_2020}]\label{ass-gbound}
	The expected squared norm of stochastic gradients 
	is uniformly bounded, 
	thus for all $ i=1, \cdots, n $ and $ r=1, \cdots, r-1$, 
	we have 
	\[\mathbb{E}\left\|\nabla F_{i}
	\left(\mathbf{w}_{t}^{i}, b_i\right)
	\right\|^{2} \leq G^{2}\]
\end{assumption}

\Cref{ass-smooth,ass-ustrongly,ass-gradientbound,ass-gbound} 
are crucial for analyzing the convergence of \emph{BFL-MEC}, 
as established in several prior works. 
These assumptions all impose constraints on the underlying loss function, 
specifying that it must not vary too quickly (\Cref{ass-smooth}) 
or too slowly (\Cref{ass-ustrongly}), 
while also bounding the magnitude of the gradients (\Cref{ass-gradientbound,ass-gbound}). 
By leveraging these assumptions and approximating $F-F^*$ with $w-w^*$, 
we can provide a convergence guarantee, 
as formalized in \Cref{THEOREM1}.

\begin{theorem}\label{THEOREM1}
	Given \Cref{ass-smooth,ass-ustrongly,ass-gradientbound,ass-gbound} hold, 
	and $\tau_{a v g}=\frac{T+1}{T+\left|C_{T}\right|-1} \bar{\tau}_{C} \stackrel{T>\left|C_{T}\right|}{=} \mathcal{O}\left(\bar{\tau}_{C}\right)$,
	we have  
	\begin{equation}\label{eq-theorem1}
		\frac{1}{T+1} \sum_{t=0}^{T}\left\|\nabla f\left(\mathrm{x}^{(t)}\right)\right\|_{2}^{2} \leq \varepsilon 
	\end{equation}
	after  $\mathcal{O}\left(\frac{\sigma^{2}}{\varepsilon^{2}}+\frac{\zeta^{2}}{\varepsilon^{2}}+\frac{\tau_{a v g} G}{\varepsilon^{\frac{3}{2}}}+\frac{\tau_{\text {avg }}}{\varepsilon}\right)$ 
	global aggregations.
\end{theorem}

According to \Cref{eq-theorem1}, 
the convergence of \emph{BFL-MEC} is guaranteed by the fact that 
the distance between the actual model $F$ and the optimal model $F^*$ 
decreases with an increasing number of global aggregations. 
Unlike other methods, 
we do not assume that the data is identically and independently distributed (IID), 
making our method robust against variations in the data distribution. 

The proof of \Cref{THEOREM1} is presented in \Cref{proof-theorem1}, 
and we further support it with experimental results in \Cref{sec-validation}.

\section{Approximate Performance of \emph{BFL-MEC}}\label{sec-summary}
In this section, 
we look at each step in BFL-MEC to analyze the performance. 
However, analyzing the overall system is difficult due to the asynchronous design, 
but we can analyze the approximate performance from the client side and edge node side separately, 
as shown in the figure below.

\begin{figure}[htbp]
	\centering
	\includegraphics[width=\linewidth]{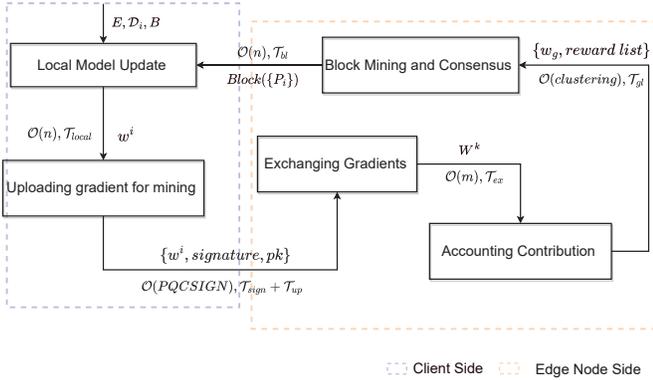}
	\caption{Approximate Performance of \emph{BFL-MEC}}
	\label{fig:delay}
\end{figure}

\subsection{Analysis of the client side}
\textbf{Local Model Update.} 
The time complexity of calculating \Cref{eq3} is $\mathcal O(E*\frac{{\mathcal D}_i}{B})$, 
as it can be calculated $\frac{{\mathcal D}i}{B}$ times 
with a batch size of $B$. 
The delay $\mathcal T_{local}$ is defined as the calculation time of this step. 
However, 
it's important to note that $E$ and $B$ are 
typically set as small constants for all clients, 
so the time complexity of \Cref{eq3} is linear and can be expressed 
as $\mathcal O(n)$ under normal circumstances.

\textbf{Uploading the gradient for mining.} 
The only computation here is the generation of the digital signature, 
so the complexity is that of the proposed procedure $PQCSIGN$, 
which we write as $\mathcal{O}(PQCSIGN)$. 
The time required to generate a digital signature for the 
same local gradient differs significantly between schemes, 
resulting in an essential source of performance discrepancies. 
The time required to generate a digital signature 
is denoted as the latency $\mathcal{T}_{sign}$. 
\Cref{sec-general} compares the latency of 
signing for various schemes in detail.
In addition, 
clients are frequently on the move, 
and ensuring the quality of the channel can be challenging. 
Moreover, 
external interferences may cause additional delays. 
Given these factors, 
we consider communication time as the primary source of delay in this phase, 
which we denote as $\mathcal{T}_{up}$. 
Consequently, 
the overall latency of this stage can be expressed as the sum of the time required for signing, 
$\mathcal{T}_{sign}$, 
and the communication time, $\mathcal{T}_{up}$, 
i.e., $\mathcal{T}_{sign}+\mathcal{T}_{up}$. 

In summary, 
the overall complexity on the client side is close to $\mathcal{O}(n)$ 
and the overall latency is $\mathcal{T}_{local}+\mathcal{T}_{sign}+\mathcal{T}_{up}$. 
Where $\mathcal{T}_{local}$ is related to the computational capability of the client itself, 
and $\mathcal{T}_{up}$ is related to the connection speed. 
Therefore $\mathcal{T}_{sign}$ becomes the most important piece of the chain, 
and in mobile edge computing scenarios, 
we have minimize $\mathcal{T}_{sign}$ while ensuring post-quantum security.

\subsection{Analysis of the edge nodes side}
\textbf{Exchanging Gradients.} 
This step involves basic communication and data exchange between edge nodes, 
and the time complexity is linear $\mathcal O(m)$. 
The edge nodes perform three main tasks: 
\begin{inparaenum}[i)]
	\item broadcasting their local gradient sets, 
	\item receiving the gradient sets from other edge nodes,
	\item adding the local gradients they do not own. 
\end{inparaenum}
The time required for all edge nodes to have the same gradient set 
is denoted as $\mathcal{T}_{ex}$, 
which is generally insignificant as long as the communication 
among the edge nodes is good, 
as required for practical applications.

\textbf{Accounting Client's Contribution.} 
The time complexity of this step depends on the clustering algorithm 
used in \Cref{alg:cii} and can be represented as $\mathcal O(clustering)$. 
During this step, 
edge nodes only need to execute the algorithm, 
so the time required depends on the efficiency of the clustering approach. 
We denote the total delay of this step as $\mathcal T_{gl}$.

\textbf{Fair Aggregation.}
The edge nodes use \Cref{eq:contriAVG} to compute the final global gradient 
with safely negligible complexity and latency. 

\textbf{Block Mining and Consensus.} 
The calculation of the hash value in accordance with \Cref{eq5} 
involves multiple iterations until the $Target$ value is reached, 
and the time required for each iteration is 
proportional to the size of the block. 
Therefore, 
the time complexity of this step can be represented as $\mathcal O(n)$, 
and we refer to the total time required for this calculation as $\mathcal T_{bl}$. 
This delay can be significant compared to other steps in the overall process.

Based on the above discussion, 
the overall complexity on the edge node side is likewise close to $\mathcal{O}(n)$, 
while the overall latency is $\mathcal{T}_{ex}+\mathcal{T}_{gl}+\mathcal{T}_{bl}$. 
Where $\mathcal{T}_{ex}$ can be reduced to a minimum level by 
good connectivity between edge nodes (e.g., wired Ethernet). 
$\mathcal{T}_{gl}$ is related to the clustering algorithm employed, 
which is a trade-off between latency and accuracy 
that should be taken as optimal according to the demand. 
$\mathcal{T}_{bl}$ is derived from a typical blockchain component, 
thus, an optimal point can be found(e.g., applying the techniques similar to that of~\cite{pokhrel_federated_2020}).

\section{Evaluation and Discussion}\label{sec-validation}
This section details the comprehensive experiments performed 
to evaluate the performance of \emph{BFL-MEC} on the real dataset. 
We varied the parameters to observe the changes in performance and delay 
under various conditions. 
Additionally, we present some novel insights, 
such as security and cost-effectiveness.

In our experiments, with important insights from~\cite{pokhrel_federated_2020, Xu2023},
we compared the performance of \emph{BFL-MEC} with three baseline methods, 
including the Blockchain, \emph{FedAvg}~\cite{mcmahanCommunicationEfficientLearningDeep2017}, 
and \emph{FedProx}~\cite{li_federated_2020}, 
on the \emph{MNIST}~\cite{lecun_gradient-based_1998} benchmark dataset. 
We evaluated the performance using the average accuracy metric, 
which is computed as $\sum\limits_{i = 1}^n {ac{c_i}/n}$, 
where $acc_i$ is the verification accuracy of client $C_i$. 
We assumed non-IID data distribution and set the default parameters 
as $n=100$, $m=2$, $\eta=0.01$, $E=5$, $B=10$, and $base=100$.

\subsection{Performance Impact}\label{sec-delay}
In our experiments, 
we set a convergence criteria for the model. 
Specifically, we define the model as converged when the change in accuracy is within $0.5\%$ 
for $5$ global aggregations. 
Moreover, 
we track the global aggregations $100$ times by default.

\subsubsection{General analysis of latency and performance}\label{sec-general}

With the fully asynchronous design, 
the edge nodes and clients are fully autonomous according to the defined policies 
so that the system latency depends mainly on 
data arrival, cryptographic processes, and data transmission. 
Here, 
we assume that the data transmission is sound, 
which is one of the advantages of edge computing. 
Also, 
to simplify the problem, data arrives continuously on each client. 
Therefore, 
only the additional delay due to the cryptographic process needs to be compared. 
The baseline is the vanilla scheme RSA (adopted by the mainstream blockchain) 
and the PQC candidates FALCON and Rainbow from NIST. 
The results are shown in \Cref{tab:delay}.

\begin{table}[htbp]
	\centering
	\caption{Delay Comparison}
	\resizebox{\linewidth}{!}{
		\begin{tabular}{llccc}
			\toprule
			\multicolumn{2}{c}{\textbf{Schemes}} & \multicolumn{1}{p{3.375em}}{\textbf{Keygen\newline{}(ms)}} & \multicolumn{1}{p{2.19em}}{\textbf{Sign\newline{}(ms)}} & \multicolumn{1}{p{2.69em}}{\textbf{Verify\newline{}(ms)}} \\
			\midrule
			Vanilla & RSA (PKCS1 v1.5) & 123   & 248   & 106 \\
			\midrule
			\multicolumn{1}{l}{\multirow{5}[9]{*}{Post-Quantum\ Safe}} & FALCON 512 & 30    & 108   & 93 \\
			\cmidrule{2-5} & FALCON 1024 & 46    & 109   & 108 \\
			\cmidrule{2-5} & Rainbow 5 Classic & 4300  & 140   & 154 \\
			\cmidrule{2-5} & Rainbow 5 Cyclic & 4620  & 120   & 219 \\
			\cmidrule{2-5} & Ours  & \textcolor[rgb]{ .929,  .49,  .192}{15*} & \textcolor[rgb]{ .929,  .49,  .192}{85*} & \textcolor[rgb]{ .929,  .49,  .192}{83*} \\
			\bottomrule
		\end{tabular}%
	}
	\label{tab:delay}%
\end{table}%

It can be seen that RSA, 
as a widely used vanilla solution, 
can sign and verify digital signatures at a good speed, 
but it takes more time to sign the local gradient. 
In post-quantum safe schemes, 
compared with RSA,
Falcon has a similar signature verification speed 
while generating a key pair and digital signature faster. 
Rainbow is only faster than RSA in signing, 
and it takes much time to initialize the key pair. 
Because Rainbow pursues the best signature size. 
BFL-MEC is leading in all comparisons. 
For generating key pairs, 
BFL-MEC is \textbf{8x} faster than RSA and \textbf{2x} faster than Falcon. 
For the signature local gradient, 
our method is \textbf{3x} faster than RSA, 
\textbf{1.5x} faster than Falcon and Rainbow. 
For verification signatures, 
BFL-MEC also outperforms all baselines. 
In conclusion, 
with the well-designed cryptographic approach, 
BFL-MEC can achieve lower latency while guaranteeing post-quantum safety. 
Hence, 
BFL-MEC is more competitive for latency-sensitive computing scenarios 
such as mobile edge computing.

\subsubsection{Impact of Thresholds}

BFL-MEC follows a fully asynchronous and autonomous design, 
where the clients and edge nodes perform the computation independently. 
In this context, 
$N$ defines when clients perform model updates and upload local gradients, 
whereas $\phi$ defines when edge nodes compute the federated model. 
These two parameters jointly shape the whole learning process. 
Thus, 
it is necessary to explore their impacts. 
Here, 
we set $N \in [50,75,100]$ and $\phi \in [5,10,15,20]$, 
and then explore the impacts of various combinations of $(\phi,N)$ on the average accuracy.
The results are shown as \Cref{fig:influenceofpara,fig:convergancetime}.

\begin{figure}[htbp]
	\centering
	\includegraphics[width=\linewidth]{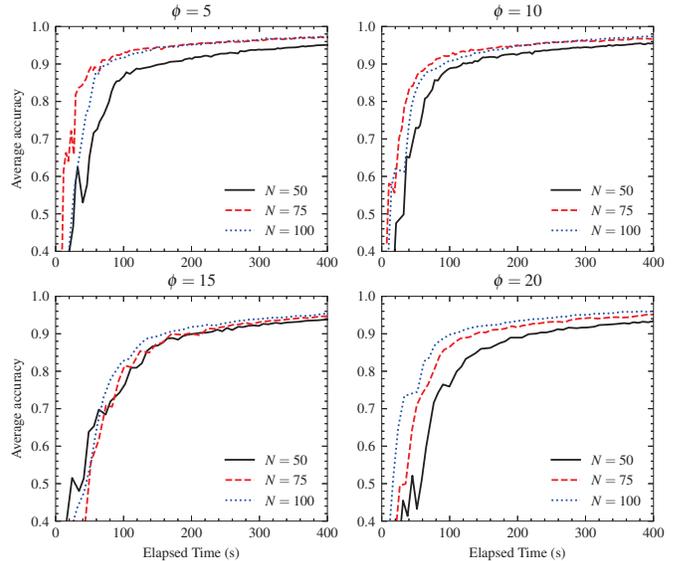}
	\caption{Accuracy of different combination}
	\label{fig:influenceofpara}
\end{figure}

We can see that when $\phi$ is set to large thresholds (e.g., $\phi=15$ and $\phi=10$), 
the accuracy of the federated model increases as $N$ increases, 
and this trend becomes more pronounced the larger $\phi$ is, 
such as when $\phi=20$. 
This is due to the fact that more data are used to update the local model, 
and the federated model is computed using more local gradients, 
amplifying such a data advantage. 
However, 
it is worth noting that the convergence of the federated model 
can be slower if $\phi$ is set to a significant value, 
as when comparing $\phi=15$ with $\phi=5$. 
The reason could be that in the case of continuous arrivals, 
the clients have to take more time to collect new data from the environment.
Meanwhile, 
the slower federated model aggregation hinders local models 
from learning global knowledge in time.

\begin{figure}[htbp]
	\centering
	\includegraphics[width=0.8\linewidth]{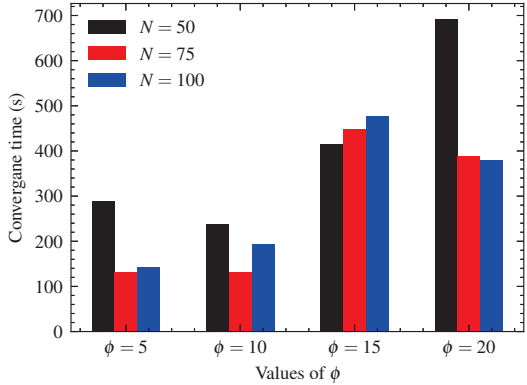}
	\caption{Convergence time of different combination}
	\label{fig:convergancetime}
\end{figure}

Interestingly, 
when $\phi$ is set to small thresholds, 
such as $\phi=5$ and $\phi=10$, 
a larger $N$ does not necessarily result in higher accuracy. 
It can be seen that in this case, 
$N=75$ and $N=100$ eventually have almost the same accuracy of the federated model, 
yet the case of $N=75$ has higher accuracy and faster convergence speed in the early learning process. 
This is caused by the fact that 
a larger $N$ reduces the rate of local gradient generation 
while the federated model is iterated quickly. 
However, 
it is important to point out that setting $N$ to a minimal value (e.g., $N=50$), 
is also considered inappropriate. 
In this case, 
the federated model has the lowest accuracy and the slowest convergence speed. 
The intuition behind that is a small $N$, 
although it enables more local gradients to be uploaded simultaneously, 
impairs the learning effect due to the small training data size used by the client. 
That is, 
the negative effect of insufficient data suppresses the 
advantage of the fast distribution of global knowledge.

In conclusion, 
it is clear that there is a trade-off between $\phi$ and $N$ and 
thus the optimal combination $(\phi,N)$ can be found. 
In the above experiment, 
the optimal combination is $(5,75)$ (see \Cref{fig:influenceofpara,fig:convergancetime}). 
To this end, we have the following insight.

\textbf{Insight 1:} 
Small $\phi$ leads to faster global knowledge transfer between edge nodes and clients, 
which helps to improve the performance of learning process. 
$N$ will affect the number of local gradients generated over the same time, 
but small $N$ will impair accuracy. 
Thus there is a trade-off. 
We can find the optimal combination to obtain the best accuracy and convergence speed.

\subsection{Cost-effectiveness}\label{sec-economy}
In this analysis, 
we investigate the impact of employing \Cref{alg:cii} 
with the discarding strategy on the accuracy 
and convergence rate of \emph{BFL-MEC}.
We use \emph{DBSCAN} as a sample, 
which is a density-based clustering method with the following advantages: 
\begin{inparaenum}[i)]
	\item It can find clusters of any shape without requiring the data set to be convex.
	\item Outliers will not significantly affect the clustering results but be discovered.
	\item It does not rely on hyperparameters such as initial values, 
	thus it can directly find the distance difference 
	between each local update and the global gradient through the clustering results.
\end{inparaenum}
We would like to emphasize that any appropriate clustering algorithm could be applied here, 
not only the one used in our implementation. 
It is worth noting that \emph{FedProx} also excludes clients to 
enhance both the convergence speed and model accuracy.
However, 
while \emph{FedProx} removes stragglers to prevent global model skewness, 
our approach excludes low-contributing clients based on the results of the clustering algorithm. 
To demonstrate the effectiveness of our contribution-based incentive mechanism, 
we consider the \emph{FedProx} with a $drop\ percent$ of $0.02$ 
as a new baseline and compare its performance to \emph{BFL-MEC}.

\begin{figure}[htbp]
	\centering
	\begin{subfigure}{0.48\linewidth}
		\includegraphics[width=\linewidth]{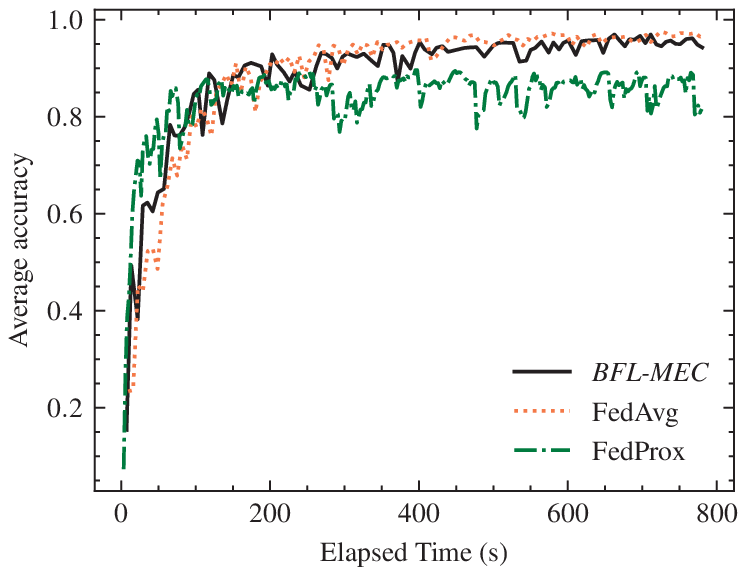}
		\caption{Without discarding}
		\label{fig:general-acc}
	\end{subfigure}
	\begin{subfigure}{0.48\linewidth}
		\includegraphics[width=\linewidth]{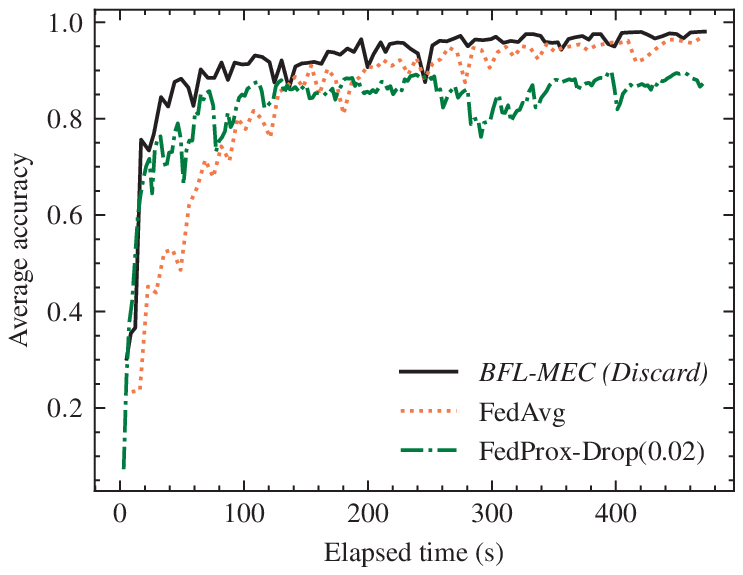}
		\caption{Discarding strategy}
		\label{fig:fast-acc}
	\end{subfigure}
	\caption{\emph{BFL-MEC} is faster without reducing accuracy}
	\label{fig:fast}
\end{figure}

Then in \Cref{fig:fast}, 
we can see the effect of the discarding strategy 
on the accuracy of \emph{BFL-MEC}.
From \Cref{fig:general-acc}, 
we can see that \emph{BFL-MEC} achieves a model performance 
almost equivalent to that of \emph{FedAvg}. 
However, \emph{FedProx} falls behind in terms of accuracy, 
and its accuracy continues to fluctuate even after convergence, 
which is due to its use of an inexact solution to accelerate convergence. 
In contrast, \emph{BFL-MEC} leverages the discard strategy to 
ensure higher accuracy and faster convergence of the model, 
as demonstrated in \Cref{fig:fast-acc} 
where the \emph{BFL-MEC} line consistently outperforms 
the \emph{FedAvg} and original \emph{BFL-MEC} lines, 
reaching the convergence point between $250$ and $300$ seconds. 
While \emph{FedProx} initially converges better than the other methods, 
its accuracy plateaus around $84\%$, which is lower than the others.

The superior performance of \emph{BFL-MEC} can be attributed to the fact 
that low-contributing clients are excluded from global aggregation, 
thereby reducing noise from low-quality data and effectively 
preventing the global model from getting stuck in local optima. 
By removing such clients, 
\emph{BFL-MEC} improves the accuracy and overall performance of the model. 
Therefore, 
as discussed above and shown in \Cref{fig:fast}, 
the discard strategy significantly improves accuracy and reduces 
the time required to reach convergence. 
Our results have shown that \emph{BFL-MEC} is faster 
and more efficient than other FL methods.

\textbf{Insight 2:} 
The implementation of the discarding strategy can lead to faster model convergence 
and higher accuracy in large-scale scenarios.

\subsection{Security by Design}\label{sec-security}
In the previous section, 
we demonstrated that using the discarding strategy in \Cref{alg:cii} 
can effectively improve the performance of \emph{BFL-MEC}. 
In this section, we evaluate the security of the proposed 
algorithm by examining its robustness against malicious clients. 
We simulate an attack scenario where some clients intentionally 
modify their local gradients to skew the global model.
Specifically, 
during each upload of the local gradient, 
the malicious clients significantly increase or 
decrease the value of the actual gradient in some direction, 
expecting to perform membership inference attack by 
analyzing the resulting change in the federated model. 
It is a typical privacy attack technique proposed in the study of \cite{nasr_comprehensive_2019}. 
We also use \emph{DBSCAN} as an example to identify 
variations in the contributions of different clients.

In experiments,
There are $10$ indexed clients $C_i\in[1,2,...,9,10]$, 
and we observe the time elapsed for the federated model to aggregate ten times.
We consider the following two attack cases:
\begin{inparaenum}[i)]
	\item crafty attackers use backdoors or Trojans to control clients 
	to perform the membership inference attacks described above 
	and are good at disguising themselves. 
	That means they may constantly be changing the client used.
	\item The attackers are the clients themselves, 
	who are curious about the data of other clients and therefore 
	perform the same privacy attacks, 
	but at the same time, 
	as participants in the system, 
	they also want to be rewarded normally.
\end{inparaenum}
To this end,
For the first case, 
we randomly designate $3$ clients as malicious clients and 
are set to become honest once a malicious client is detected. 
Meanwhile, 
another client will become malicious. 
But the total number of malicious clients is constant.
Then, we observe the detection rate of malicious clients.
As for another case,
we specify that clients $C_i\in[6,8,9]$ is curious about 
the data of other participants and this situation will not change. 
Then, 
we observe the cumulative rewards obtained by all clients 
during federated model aggregation.
The results are as shown in \Cref{fig:safe}.

\begin{figure}[htbp]
	\centering
	\begin{subfigure}{0.48\linewidth}
		\includegraphics[width=\linewidth]{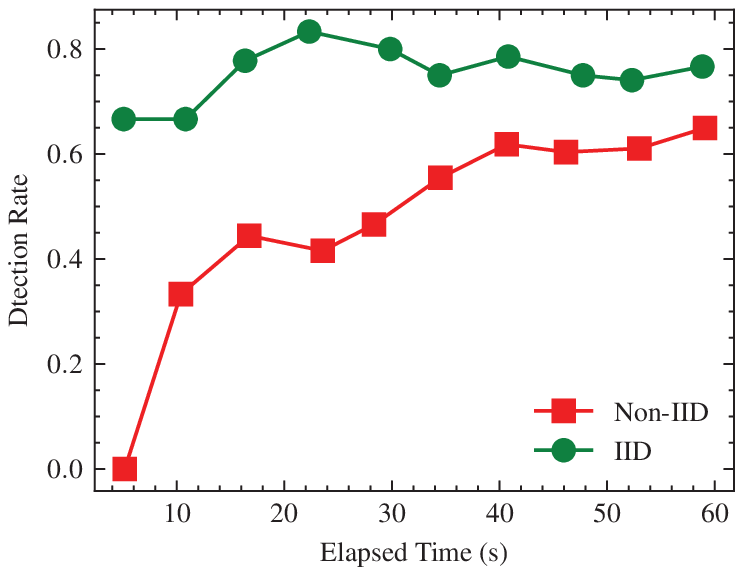}
		\caption{Attack detection rates}
		\label{fig:safe-attack}
	\end{subfigure}
	\begin{subfigure}{0.48\linewidth}
		\includegraphics[width=\linewidth]{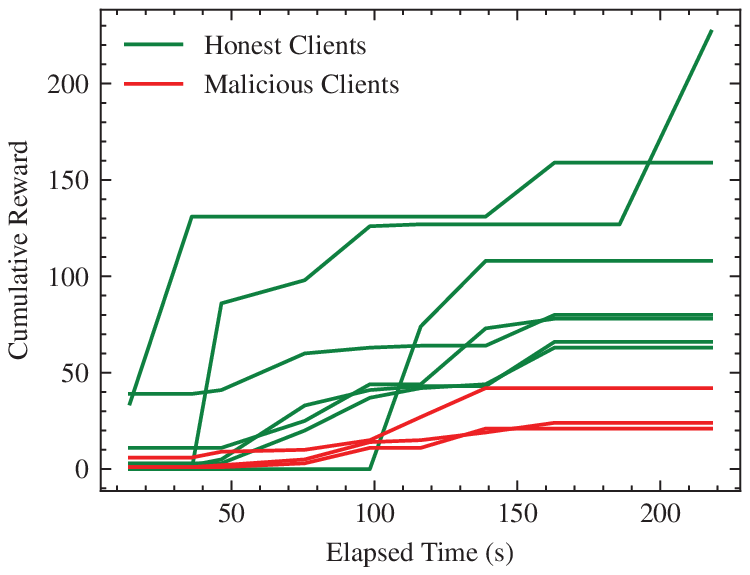}
		\caption{Clients' cumulative reward}
		\label{fig:safe-reward}
	\end{subfigure}
	\caption{\emph{BFL-MEC} is faster without reducing accuracy}
	\label{fig:safe}
\end{figure}

From \Cref{fig:safe-attack}, 
we can see that when the data distribution of clients satisfies IID, 
the detection rate of malicious attacks is at a high level from 
the very beginning (about 67\%) and eventually stabilizes at close to 80\%.
This implies that, given that the vast majority of clients are honest, 
the impact of malicious clients can be clearly detected 
as their modified local gradients are significantly 
different from the normal ones.
In the case of IID, 
such differences are more easily identified.
On the contrary, 
if the data distribution of clients is heterogeneous, 
the attacker's forged local gradients may successfully escape detection 
by the mechanism at the beginning of the learning process.
Nevertheless, 
with a few global aggregations, 
the detection rate will boost quickly and 
eventually reach the optimal level (about 66\%). 
To put it differently, 
the detection rate of malicious nodes is observed 
to increase with the convergence of the model. 
This is because, 
as the model converges, 
the local gradients become more similar to each other, 
and the modified gradients used by malicious clients are more distant from the normal ones. 
Thus, 
the clustering algorithm can more effectively detect the anomalies. 
In the case of IID, 
the distribution of data is good, 
and normal gradients are more spatially concentrated, 
making it easier to detect anomalies from the beginning. 
However, 
in the case of non-IID, 
the detection rate can also be achieved gradually as the model 
converges. 
These results indicate that \emph{BFL-MEC} can resist malicious attacks effectively, 
whether in the case of IID or non-IID data distributions.

From \Cref{fig:safe-reward}, 
clients identified by the discard strategy as low contributing will not be rewarded, 
so that the cumulative reward curve will appear as a straight line. 
We can see that many local gradients of honest clients are also discarded at the beginning, 
and only a few honest clients are rewarded. 
However, 
as long as they remain honest, 
all of them will be rewarded after a few global aggregations and 
will end up with a significantly higher reward than malicious clients. 
Such a result is caused by Non-IID, 
which coincides with \Cref{fig:safe-attack}. 
Thus the effectiveness of the discard strategy is proved. 
Although some malicious clients survive the discarding strategy, 
the fair aggregation will still play a role as the second guarantee. 
In such a case, 
the rewards earned by malicious clients will be minimal, 
and the entire cumulative reward curve barely rises.
In addition, 
the final rewards obtained are different for honest clients, 
thus capturing the differences in contributions. 
This means that \Cref{alg:cii} achieves the goal of 
personalized incentives very well.

In summary, 
\emph{BFL-MEC} has incorporated multiple design 
aspects to ensure the security and privacy of the system:
\begin{inparaenum}[i)]
	\item \textbf{Post-Quantum Security.} 
	We employ a post-quantum 
	secure design to sign local gradients 
	to avoid interception and malicious tampering 
	during data transmission (see \Cref{fig:rsa}). 
	\item \textbf{Secure Data Sharing.} 
	The use of blockchain technology ensures immutability of the data, 
	making it difficult for attackers to modify or delete the recorded data.
	\item \textbf{Black-box attack resistance.} 
	We adopt \Cref{alg:cii} to uncover the variance in contribution 
	among clients and remove the low-contributing local gradients, 
	which could potentially be forged by malicious attackers, 
	to enhance the system's resilience to attacks.
	Fair aggregation as a second guarantee ensures 
	that malicious clients cannot get high rewards.
	\item \textbf{White-box attack resistance.} 
	To ensure the privacy and security of the data, 
	the blockchain doesn't record original local gradients on. 
	This prevents clients from accessing and 
	exploiting this sensitive information (see \Cref{pqc}).
\end{inparaenum}
Thus, 
\emph{BFL-MEC} provides the privacy and security guarantee by design
for the whole system dynamics.

\section{Conclusion} \label{sec-conclusions}
In this paper, 
we develop a blockchain-based federated learning framework 
for mobile edge computing, called BFL-MEC, 
which aims to provide the foundation for future paradigms 
in this field such as autonomous vehicles, 
mobile crowd sensing, and metaverse. 
We propose several mechanisms to address the performance, 
motivational, and security challenges faced by vanilla BFL. 
First, we propose a fully asynchronous design in which the 
client and edge nodes have independent working policies 
that can tolerate connection loss due to mobility in the MEC network. 
Second, we use a signature scheme based on post-quantum cryptography, 
which makes the proposed framework resilient to threats from quantum 
computers and reduces the latency of the signature-verification process. 
Third, 
we propose a contribution-based incentive mechanism that 
allows the use of multiple clustering algorithms to discover 
differences in client contributions and then assign personalized rewards. 
Finally, 
we developed a fair aggregation mechanism that computing 
the global model based on the weight of the client's contribution 
minimizes the impact of attacks from the clients themselves. 

\appendix
\section{Proof of THEOREM \ref{THEOREM1}}\label{proof-theorem1}
As mentioned, 
we denote the last global aggregation step $t$, 
and us the notation $\mathbb{E}(\cdot)$ as the expectation.
Using the perturbed iterative technique~\cite{maniaPerturbedIterateAnalysis2017}, a virtual sequence
\begin{equation}
	\tilde{\mathbf{x}}^{(0)}=\mathrm{x}^{(0)} \quad \tilde{\mathbf{x}}^{(t+1)}=\tilde{\mathbf{x}}^{(t)}-\eta \nabla F_{k_{t}}\left(\mathrm{x}^{(t)}, b_{t+\hat{\tau}_{t}}\right),
\end{equation}
can be defined with $\hat{\tau}_{t}$ as the delay involved in the gradient computation. 

Next, we have the following two lemmas: \Cref{LEMMA1} and \Cref{LEMMA2}.

\begin{lemma}[Descent Lemma]\label{LEMMA1}
	Given \Cref{ass-smooth,ass-ustrongly,ass-gradientbound} hold,
	\begin{equation*}
		\begin{aligned}
			\mathbb{E}_{t+1} f\left(\tilde{\mathrm{x}}^{(t+1)}\right) &\leq f\left(\tilde{\mathrm{x}}^{(t)}\right)-\frac{\eta}{4}\left\|\nabla f\left(\mathrm{x}^{(t)}\right)\right\|_{2}^{2}+\frac{L \eta^{2} \sigma^{2}}{2} \\ 
			&+L \eta^{2} \zeta^{2}+\frac{\eta L^{2}}{2}\left\|\mathrm{x}^{(t)}-\tilde{\mathrm{x}}^{(t)}\right\|_{2}^{2},
		\end{aligned}
	\end{equation*}
\end{lemma}
if and only if $\eta_{t} \leq \frac{1}{4L}$.

\begin{proof}
	Because the function $f$ is L-smooth (\Cref{ass-smooth}), 
	and taking an expectation over both the batch $b$ 
	and index $C_t$ of ``activated client'', we have 
	\begin{equation*}
		\begin{aligned}
			\mathbb{E}_{t+1} f\left(\tilde{\mathbf{x}}^{(t+1)}\right) 
			&=\mathbb{E}_{t+1} f\left(\tilde{\mathbf{x}}^{(t)}-\eta \nabla F_{k_{t}}\left(\mathbf{x}^{(t)}, b_{t+t_{t}}\right)\right) \\ 
			&\leq f\left(\tilde{\mathbf{x}}^{(t)}\right)-\eta \underbrace{\left\langle\nabla f\left(\tilde{\mathbf{x}}^{(t)}\right), \nabla f\left(\mathbf{x}^{(t)}\right)\right\rangle}_{=: Term_{1}} \\ 
			&+\mathbb{E}_{t+1} \frac{L}{2} \eta^{2} \underbrace{\left\|\nabla F_{k_{\mathrm{t}}}\left(\mathbf{x}^{(t)}, b_{t+t_{\mathrm{t}}}\right)\right\|_{2}^{2}}_{=: Term_{2}}.
		\end{aligned}
	\end{equation*}
	where $Term_{1}$
	\begin{equation*}
		\begin{aligned}
			& =-\frac{\eta}{2}\left\|\nabla f\left(\mathrm{x}^{(t)}\right)\right\|^{2}-\frac{\eta}{2}\left\|\nabla f\left(\tilde{\mathrm{x}}^{(t)}\right)\right\|^{2} \\ 
			&+\frac{\eta}{2}\left\|\nabla f\left(\mathrm{x}^{(t)}\right)-\nabla f\left(\tilde{\mathrm{x}}^{(t)}\right)\right\|^{2} \\ & \leq-\frac{\eta}{2}\left\|\nabla f\left(\mathrm{x}^{(t)}\right)\right\|^{2}+\frac{\eta}{2}\left\|\nabla f\left(\mathrm{x}^{(t)}\right)-\nabla f\left(\tilde{\mathrm{x}}^{(t)}\right)\right\|^{2}
		\end{aligned}
	\end{equation*}
	and $Term_{2}$
	\begin{equation*}
		\begin{aligned}
			&=\mathbb{E}_{t+1}\left\|\nabla F_{k_{t}}\left(\mathrm{x}^{(t)}, b_{t+\hat{\tau}_{t}}\right) \pm \nabla f_{j_{t}}\left(\mathrm{x}^{(t)}\right) \pm \nabla f\left(\mathrm{x}^{(t)}\right)\right\|_{2}^{2} \\ 
			&\leq \sigma^{2}+2\mathbb{E}_{k_{t}}\left\|\nabla f_{k_{t}}\left(\mathrm{x}^{(t)}\right)-\nabla f\left(\mathrm{x}^{(t)}\right)\right\|_{2}^{2}+2\left\|\nabla f\left(\mathrm{x}^{(t)}\right)\right\|_{2}^{2} \\ \quad 
			&\leq \sigma^{2}+2\zeta^{2}+2\left\|\nabla f\left(\mathrm{x}^{(t)}\right)\right\|_{2}^{2}
		\end{aligned}
	\end{equation*}
	Combining these two together and using \Cref{ass-smooth} to estimate $\left\|\nabla f\left(\mathrm{x}^{(t)}\right)-\nabla f\left(\tilde{\mathbf{x}}^{(t)}\right)\right\|_{2}^{2}$, 
	we have 
	\begin{equation*}
		\begin{aligned}
			\mathbb{E}_{t+1} f\left(\tilde{\mathrm{x}}^{(t+1)}\right) &\leq f\left(\tilde{\mathrm{x}}^{(t)}\right)-\left(\frac{\eta}{2}-L \eta^{2}\right)\left\|\nabla f\left(\mathrm{x}^{(t)}\right)\right\|_{2}^{2} \\
			&+\frac{\eta}{2} L^{2}\left\|\mathrm{x}^{(t)}-\tilde{\mathrm{x}}^{(t)}\right\|^{2}+\frac{L \eta^{2} \sigma^{2}}{2}+L \eta^{2} \zeta^{2}
		\end{aligned}
	\end{equation*}
	Applying $\eta \leq \frac{1}{4L}$ proves the  lemma. 
\end{proof}

To estimate the distance $\left\|\mathrm{x}^{(t)}-\tilde{\mathbf{x}}^{(t)}\right\|_{2}^{2}$, we have the following lemma.

\begin{lemma}\label{LEMMA2}
	Given \Cref{ass-smooth,ass-ustrongly,ass-gradientbound,ass-gbound} hold, 
	for $\eta_{t} \equiv \eta \leq \frac{1}{4L \tau_{C}}$, 
	
	\begin{equation*}
		\begin{aligned}
			\frac{1}{T+1} \sum_{t=0}^{T} \mathbb{E}\left\|\mathrm{x}^{(t)}-\tilde{\mathrm{x}}^{(t)}\right\|_{2}^{2} \leq \frac{\eta \sigma^{2}}{4L}+\eta^{2} \tau_{C}^{2} G^{2}
		\end{aligned}
	\end{equation*}
\end{lemma}

\begin{proof}
	Using the steps as above, we can obtain 
	\begin{equation*}
		\begin{aligned}
			\mathbb{E}\left\|\mathrm{x}^{(t)}-\tilde{\mathbf{x}}^{(t)}\right\|_{2}^{2} 
			& =\mathbb{E} \eta^{2}\left\|\sum_{i \in \mathcal{C}_{t}} \nabla F_{j_{i}}\left(\mathbf{x}^{(i)}, b_{i+\hat{\tau}_{i}}\right)\right\|_{2}^{2} \\
			&\leq \eta^{2} \tau_{C} \sigma^{2}+\eta^{2} \mathbb{E} \sum_{i \in \mathcal{C}_{t}} \nabla f_{j_{i}}\left(\mathrm{x}^{(i)}\right) \|_{2}^{2} \\ 
			&\leq \eta^{2} \tau_{C} \sigma^{2}+\eta^{2} \tau_{C} \sum_{i \in \mathcal{C}_{t}} \mathbb{E}\left\|\nabla f_{j_{i}}\left(\mathrm{x}^{(i)}\right)\right\|_{2}^{2} \\ 
			&\leq \eta^{2} \tau_{C} \sigma^{2}+\eta^{2} \tau_{C}^{2} G^{2} \\ 
			&\leq \frac{\eta \sigma^{2}}{4L}+\eta^{2} \tau_{C}^{2} G^{2},
		\end{aligned}
	\end{equation*}
	where we applying $\eta \leq \frac{1}{4L \tau_{C}}$ on the last line proves the lemma.
\end{proof}

Next, we provide the proof of \Cref{THEOREM1}. 

\begin{proof}
	First, compute the average of \Cref{LEMMA1}, as follows
	\begin{equation*}
		\begin{aligned}
			\frac{1}{T+1} \sum_{t=0}^{T} \mathbb{E}\left\|\nabla f\left(\mathrm{x}^{(t)}\right)\right\|_{2}^{2} 
			&\leq \frac{4}{\eta(T+1)}\left(f\left(\mathrm{x}^{0}\right)-f\left(\mathrm{x}^{T}\right)\right) \\ 
			+2L \eta \sigma^{2}+4L \eta \zeta^{2}  
			&+\frac{2L^{2}}{T+1} \sum_{t=0}^{T} \mathbb{E}\left\|\mathrm{x}^{(t)}-\tilde{\mathrm{x}}^{(t)}\right\|_{2}^{2}.
		\end{aligned}
	\end{equation*}
	Then plug the results of \Cref{LEMMA2} into above equation 
	\begin{equation*}
		\begin{aligned}
			\frac{1}{T+1} \sum^{T} \mathbb{E}\left\|\nabla f\left(\mathrm{x}^{(t)}\right)\right\|_{2}^{2} 
			&\leq \frac{4}{\eta(T+1)}\left(f\left(\mathrm{x}^{0}\right)-f\left(\mathrm{x}^{T}\right)\right) \\ 
			&+3L \eta \sigma^{2}+4L \eta \zeta^{2}+2L^{2} \eta^{2} \tau_{C}^{2} G^{2}
		\end{aligned}
	\end{equation*}
	Tuning the stepsize of \Cref{LEMMA1} along the lines of \cite{koloskovaUnifiedTheoryDecentralized2020}, proves \Cref{THEOREM1}.
\end{proof}

\bibliography{fair-bfl}

\begin{thebibliography}{10}
\providecommand{\url}[1]{#1}
\csname url@samestyle\endcsname
\providecommand{\newblock}{\relax}
\providecommand{\bibinfo}[2]{#2}
\providecommand{\BIBentrySTDinterwordspacing}{\spaceskip=0pt\relax}
\providecommand{\BIBentryALTinterwordstretchfactor}{4}
\providecommand{\BIBentryALTinterwordspacing}{\spaceskip=\fontdimen2\font plus
\BIBentryALTinterwordstretchfactor\fontdimen3\font minus
  \fontdimen4\font\relax}
\providecommand{\BIBforeignlanguage}[2]{{%
\expandafter\ifx\csname l@#1\endcsname\relax
\typeout{** WARNING: IEEEtran.bst: No hyphenation pattern has been}%
\typeout{** loaded for the language `#1'. Using the pattern for}%
\typeout{** the default language instead.}%
\else
\language=\csname l@#1\endcsname
\fi
#2}}
\providecommand{\BIBdecl}{\relax}
\BIBdecl

\bibitem{mcmahanCommunicationEfficientLearningDeep2017}
\BIBentryALTinterwordspacing
B.~McMahan, E.~Moore, D.~Ramage, S.~Hampson, and B.~A.~y. Arcas,
  ``\BIBforeignlanguage{en}{Communication-{Efficient} {Learning} of {Deep}
  {Networks} from {Decentralized} {Data}},'' in
  \emph{\BIBforeignlanguage{en}{Proceedings of the 20th {International}
  {Conference} on {Artificial} {Intelligence} and {Statistics}}}.\hskip 1em
  plus 0.5em minus 0.4em\relax PMLR, Apr. 2017, pp. 1273--1282, iSSN:
  2640-3498. [Online]. Available:
  \url{https://proceedings.mlr.press/v54/mcmahan17a.html}
\BIBentrySTDinterwordspacing

\bibitem{yangFLASHHeterogeneityAwareFederated2022}
C.~Yang, M.~Xu, Q.~Wang, Z.~Chen, K.~Huang, Y.~Ma, K.~Bian, G.~Huang, Y.~Liu,
  X.~Jin, and X.~Liu, ``{FLASH}: {Heterogeneity}-{Aware} {Federated} {Learning}
  at {Scale},'' \emph{IEEE Transactions on Mobile Computing}, pp. 1--18, 2022.

\bibitem{pokhrel_federated_2020}
S.~R. Pokhrel and J.~Choi, ``Federated {Learning} {With} {Blockchain} for
  {Autonomous} {Vehicles}: {Analysis} and {Design} {Challenges},'' \emph{IEEE
  Transactions on Communications}, vol.~68, no.~8, pp. 4734--4746, Aug. 2020.

\bibitem{romanFeaturesChallengesSecurity2013}
R.~Roman, J.~Zhou, and J.~Lopez, ``\BIBforeignlanguage{english}{On the features
  and challenges of security and privacy in distributed internet of things},''
  \emph{\BIBforeignlanguage{english}{Computer Networks}}, vol.~57, no.~10, pp.
  2266--2279, Jul. 2013.

\bibitem{nasr_comprehensive_2019}
M.~Nasr, R.~Shokri, and A.~Houmansadr, ``Comprehensive {Privacy} {Analysis} of
  {Deep} {Learning}: {Passive} and {Active} {White}-box {Inference} {Attacks}
  against {Centralized} and {Federated} {Learning},'' in \emph{2019 {IEEE}
  {Symposium} on {Security} and {Privacy} ({SP})}, May 2019, pp. 739--753.

\bibitem{nguyen_federated_2021}
D.~C. Nguyen, M.~Ding, Q.-V. Pham, P.~N. Pathirana, L.~B. Le, A.~Seneviratne,
  J.~Li, D.~Niyato, and H.~V. Poor, ``Federated {Learning} {Meets} {Blockchain}
  in {Edge} {Computing}: {Opportunities} and {Challenges},'' \emph{IEEE
  Internet of Things Journal}, vol.~8, no.~16, pp. 12\,806--12\,825, Aug. 2021.

\bibitem{Xu2023}
\BIBentryALTinterwordspacing
R.~Xu, S.~R. Pokhrel, Q.~Lan, and G.~Li, ``{FAIR-BFL}: Flexible and incentive
  redesign for blockchain-based federated learning,'' in \emph{Proceedings of
  the 51st International Conference on Parallel Processing}, ser. ICPP
  '22.\hskip 1em plus 0.5em minus 0.4em\relax New York, NY, USA: Association
  for Computing Machinery, 2023. [Online]. Available:
  \url{https://doi.org/10.1145/3545008.3545040}
\BIBentrySTDinterwordspacing

\bibitem{shiEdgeComputingVision2016}
W.~Shi, J.~Cao, Q.~Zhang, Y.~Li, and L.~Xu, ``Edge {Computing}: {Vision} and
  {Challenges},'' \emph{IEEE Internet of Things Journal}, vol.~3, no.~5, pp.
  637--646, Oct. 2016, conference Name: IEEE Internet of Things Journal.

\bibitem{sunAcceleratingConvergenceFederated2023}
W.~Sun, Y.~Zhao, W.~Ma, B.~Guo, L.~Xu, and T.~Q. Duong, ``Accelerating
  {Convergence} of {Federated} {Learning} in {MEC} with {Dynamic}
  {Community},'' \emph{IEEE Transactions on Mobile Computing}, pp. 1--17, 2023,
  conference Name: IEEE Transactions on Mobile Computing.

\bibitem{liuAdaptiveAsynchronousFederated2023}
J.~Liu, H.~Xu, L.~Wang, Y.~Xu, C.~Qian, J.~Huang, and H.~Huang, ``Adaptive
  {Asynchronous} {Federated} {Learning} in {Resource}-{Constrained} {Edge}
  {Computing},'' \emph{IEEE Transactions on Mobile Computing}, vol.~22, no.~2,
  pp. 674--690, Feb. 2023, conference Name: IEEE Transactions on Mobile
  Computing.

\bibitem{wittDecentralIncentivizedFederated2023}
L.~Witt, M.~Heyer, K.~Toyoda, W.~Samek, and D.~Li, ``Decentral and
  {Incentivized} {Federated} {Learning} {Frameworks}: {A} {Systematic}
  {Literature} {Review},'' \emph{IEEE Internet of Things Journal}, vol.~10,
  no.~4, pp. 3642--3663, Feb. 2023, conference Name: IEEE Internet of Things
  Journal.

\bibitem{zhangScalableLowLatencyFederated2022}
Z.~Zhang, Z.~Gao, Y.~Guo, and Y.~Gong, ``Scalable and {Low}-{Latency}
  {Federated} {Learning} {With} {Cooperative} {Mobile} {Edge} {Networking},''
  \emph{IEEE Transactions on Mobile Computing}, pp. 1--11, 2022, conference
  Name: IEEE Transactions on Mobile Computing.

\bibitem{bao_flchain_2019}
X.~Bao, C.~Su, Y.~Xiong, W.~Huang, and Y.~Hu, ``{FLChain}: {A} {Blockchain} for
  {Auditable} {Federated} {Learning} with {Trust} and {Incentive},'' in
  \emph{2019 5th {International} {Conference} on {Big} {Data} {Computing} and
  {Communications} ({BIGCOM})}, 2019, pp. 151--159.

\bibitem{sunPainFLPersonalizedPrivacyPreserving2021}
P.~Sun, H.~Che, Z.~Wang, Y.~Wang, T.~Wang, L.~Wu, and H.~Shao, ``Pain-{FL}:
  {Personalized} {Privacy}-{Preserving} {Incentive} for {Federated}
  {Learning},'' \emph{IEEE Journal on Selected Areas in Communications},
  vol.~39, no.~12, pp. 3805--3820, Dec. 2021, conference Name: IEEE Journal on
  Selected Areas in Communications.

\bibitem{luoSVFLEfficientSecure2022}
F.~Luo, S.~Al-Kuwari, and Y.~Ding, ``{SVFL}: {Efficient} {Secure} {Aggregation}
  and {Verification} for {Cross}-{Silo} {Federated} {Learning},'' \emph{IEEE
  Transactions on Mobile Computing}, pp. 1--14, 2022, conference Name: IEEE
  Transactions on Mobile Computing.

\bibitem{awan_poster_2019}
\BIBentryALTinterwordspacing
S.~Awan, F.~Li, B.~Luo, and M.~Liu, ``Poster: {A} {Reliable} and {Accountable}
  {Privacy}-{Preserving} {Federated} {Learning} {Framework} using the
  {Blockchain},'' in \emph{Proceedings of the 2019 {ACM} {SIGSAC} {Conference}
  on {Computer} and {Communications} {Security}}, ser. {CCS} '19.\hskip 1em
  plus 0.5em minus 0.4em\relax New York, NY, USA: Association for Computing
  Machinery, Nov. 2019, pp. 2561--2563. [Online]. Available:
  \url{https://doi.org/10.1145/3319535.3363256}
\BIBentrySTDinterwordspacing

\bibitem{majeed_flchain_2019}
U.~Majeed and C.~S. Hong, ``{FLchain}: {Federated} {Learning} via {MEC}-enabled
  {Blockchain} {Network},'' in \emph{2019 20th {Asia}-{Pacific} {Network}
  {Operations} and {Management} {Symposium} ({APNOMS})}, Sep. 2019, pp. 1--4.

\bibitem{lu_blockchain_2020}
Y.~Lu, X.~Huang, Y.~Dai, S.~Maharjan, and Y.~Zhang, ``Blockchain and
  {Federated} {Learning} for {Privacy}-{Preserved} {Data} {Sharing} in
  {Industrial} {IoT},'' \emph{IEEE Transactions on Industrial Informatics},
  vol.~16, no.~6, pp. 4177--4186, Jun. 2020.

\bibitem{li_blockchain-based_2021}
Y.~Li, C.~Chen, N.~Liu, H.~Huang, Z.~Zheng, and Q.~Yan, ``A
  {Blockchain}-{Based} {Decentralized} {Federated} {Learning} {Framework} with
  {Committee} {Consensus},'' \emph{IEEE Network}, vol.~35, no.~1, pp. 234--241,
  Jan. 2021.

\bibitem{kim_blockchained_2020}
H.~Kim, J.~Park, M.~Bennis, and S.-L. Kim, ``Blockchained {On}-{Device}
  {Federated} {Learning},'' \emph{IEEE Communications Letters}, vol.~24, no.~6,
  pp. 1279--1283, Jun. 2020.

\bibitem{ducasCRYSTALSDilithiumLatticeBasedDigital2018}
\BIBentryALTinterwordspacing
L.~Ducas, E.~Kiltz, T.~Lepoint, V.~Lyubashevsky, P.~Schwabe, G.~Seiler, and
  D.~Stehlé, ``\BIBforeignlanguage{en}{{CRYSTALS}-{Dilithium}: {A}
  {Lattice}-{Based} {Digital} {Signature} {Scheme}},''
  \emph{\BIBforeignlanguage{en}{IACR Transactions on Cryptographic Hardware and
  Embedded Systems}}, pp. 238--268, Feb. 2018. [Online]. Available:
  \url{https://tches.iacr.org/index.php/TCHES/article/view/839}
\BIBentrySTDinterwordspacing

\bibitem{fredrikson_model_2015}
\BIBentryALTinterwordspacing
M.~Fredrikson, S.~Jha, and T.~Ristenpart, ``Model {Inversion} {Attacks} that
  {Exploit} {Confidence} {Information} and {Basic} {Countermeasures},'' in
  \emph{Proceedings of the 22nd {ACM} {SIGSAC} {Conference} on {Computer} and
  {Communications} {Security}}, ser. {CCS} '15.\hskip 1em plus 0.5em minus
  0.4em\relax New York, NY, USA: Association for Computing Machinery, Oct.
  2015, pp. 1322--1333. [Online]. Available:
  \url{https://doi.org/10.1145/2810103.2813677}
\BIBentrySTDinterwordspacing

\bibitem{koloskovaSharperConvergenceGuarantees2022}
\BIBentryALTinterwordspacing
A.~Koloskova, S.~U. Stich, and M.~Jaggi, ``\BIBforeignlanguage{en}{Sharper
  {Convergence} {Guarantees} for {Asynchronous} {SGD} for {Distributed} and
  {Federated} {Learning}},'' in \emph{\BIBforeignlanguage{en}{Advances in
  {Neural} {Information} {Processing} {Systems} 35 ({NeurIPS} 2022)}}, New
  Orleans, USA, Nov. 2022. [Online]. Available:
  \url{https://publications.cispa.saarland/3800/}
\BIBentrySTDinterwordspacing

\bibitem{stich_local_2019}
\BIBentryALTinterwordspacing
S.~U. Stich, ``Local {SGD} converges fast and communicates little,'' in
  \emph{7th international conference on learning representations, {ICLR} 2019,
  new orleans, {LA}, {USA}, may 6-9, 2019}.\hskip 1em plus 0.5em minus
  0.4em\relax OpenReview.net, 2019. [Online]. Available:
  \url{https://arxiv.org/abs/1805.09767}
\BIBentrySTDinterwordspacing

\bibitem{li_convergence_2020}
\BIBentryALTinterwordspacing
X.~Li, K.~Huang, W.~Yang, S.~Wang, and Z.~Zhang, ``On the convergence of
  {FedAvg} on non-iid data,'' in \emph{8th international conference on learning
  representations, {ICLR} 2020, addis ababa, ethiopia, april 26-30,
  2020}.\hskip 1em plus 0.5em minus 0.4em\relax OpenReview.net, 2020. [Online].
  Available: \url{https://openreview.net/forum?id=HJxNAnVtDS}
\BIBentrySTDinterwordspacing

\bibitem{li_federated_2020}
\BIBentryALTinterwordspacing
T.~Li, A.~K. Sahu, M.~Zaheer, M.~Sanjabi, A.~Talwalkar, and V.~Smith,
  ``Federated optimization in heterogeneous networks,'' in \emph{Proceedings of
  machine learning and systems}, I.~Dhillon, D.~Papailiopoulos, and V.~Sze,
  Eds., vol.~2, 2020, pp. 429--450. [Online]. Available:
  \url{https://proceedings.mlsys.org/paper/2020/file/38af86134b65d0f10fe33d30dd76442e-Paper.pdf}
\BIBentrySTDinterwordspacing

\bibitem{dinhFederatedLearningWireless2021}
C.~T. Dinh, N.~H. Tran, M.~N.~H. Nguyen, C.~S. Hong, W.~Bao, A.~Y. Zomaya, and
  V.~Gramoli, ``Federated {Learning} {Over} {Wireless} {Networks}:
  {Convergence} {Analysis} and {Resource} {Allocation},'' \emph{IEEE/ACM
  Transactions on Networking}, vol.~29, no.~1, pp. 398--409, Feb. 2021,
  conference Name: IEEE/ACM Transactions on Networking.

\bibitem{lecun_gradient-based_1998}
Y.~Lecun, L.~Bottou, Y.~Bengio, and P.~Haffner, ``Gradient-based learning
  applied to document recognition,'' \emph{Proceedings of the IEEE}, vol.~86,
  no.~11, pp. 2278--2324, Nov. 1998.

\bibitem{maniaPerturbedIterateAnalysis2017}
\BIBentryALTinterwordspacing
H.~Mania, X.~Pan, D.~Papailiopoulos, B.~Recht, K.~Ramchandran, and M.~I.
  Jordan, ``Perturbed {Iterate} {Analysis} for {Asynchronous} {Stochastic}
  {Optimization},'' \emph{SIAM Journal on Optimization}, vol.~27, no.~4, pp.
  2202--2229, Jan. 2017, publisher: Society for Industrial and Applied
  Mathematics. [Online]. Available:
  \url{https://epubs.siam.org/doi/abs/10.1137/16M1057000}
\BIBentrySTDinterwordspacing

\bibitem{koloskovaUnifiedTheoryDecentralized2020}
\BIBentryALTinterwordspacing
A.~Koloskova, N.~Loizou, S.~Boreiri, M.~Jaggi, and S.~Stich,
  ``\BIBforeignlanguage{en}{A {Unified} {Theory} of {Decentralized} {SGD} with
  {Changing} {Topology} and {Local} {Updates}},'' in
  \emph{\BIBforeignlanguage{en}{Proceedings of the 37th {International}
  {Conference} on {Machine} {Learning}}}.\hskip 1em plus 0.5em minus
  0.4em\relax PMLR, Nov. 2020, pp. 5381--5393, iSSN: 2640-3498. [Online].
  Available: \url{https://proceedings.mlr.press/v119/koloskova20a.html}
\BIBentrySTDinterwordspacing

\end{thebibliography}
\bibliographystyle{IEEEtran}

\end{document}